\documentclass[sigplan,screen]{acmart}
\def\fullversion

\copyrightyear{2023}
\acmYear{2023}
\setcopyright{rightsretained}
\acmConference[PPoPP '23]{The 28th ACM SIGPLAN Annual Symposium on Principles and Practice of Parallel Programming}{February 25-March 1, 2023}{Montreal, QC, Canada}
\acmBooktitle{The 28th ACM SIGPLAN Annual Symposium on Principles and Practice of Parallel Programming (PPoPP '23), February 25-March 1, 2023, Montreal, QC, Canada}
\acmDOI{10.1145/3572848.3577483}
\acmISBN{979-8-4007-0015-6/23/02}

\usepackage{caption}
\usepackage{colortbl}
\usepackage{style}
\usepackage{bigstrut}
\usepackage{microtype}
\usepackage{tabularx}
\usepackage{multirow}
\usepackage{multicol}
\usepackage{rotating}
\usepackage{lipsum}
\usepackage{balance}
\usepackage{mfirstuc}
\usepackage{titlecaps}
\usepackage{listings}
\usepackage{float}
\usepackage[rightcaption]{sidecap}
\usepackage[normalem]{ulem}
\usepackage{wrapfig}

\newcommand{\iffullversion}[1]{{{\unless\ifx\fullversion\undefined #1 \fi}}}
\newcommand{\ifconference}[1]{{{\unless\ifx\conference\undefined #1 \fi}}}

\newcommand{\mathfunc}[1]{\mathit{#1}}

\newcommand{\knn}{$k$-NN\xspace}

\newcommand{\implementationname}[1]{\textsf{#1}}
\newcommand{\connectit}{\implementationname{ConnectIt}\xspace}

\newcommand{\hipc}{\implementationname{SM'14}\xspace}
\newcommand{\gbbs}{\implementationname{GBBS}\xspace}
\newcommand{\seq}{\implementationname{SEQ}\xspace}
\newcommand{\Ouralgo}{\textsf{FAST-BCC}\xspace}
\newcommand{\ouralgo}{\textsf{FAST-BCC}\xspace}

\newcommand{\first}{\mathfunc{first}\xspace}
\newcommand{\last}{\mathfunc{last}\xspace}
\newcommand{\low}{\mathfunc{low}\xspace}
\newcommand{\high}{\mathfunc{high}\xspace}

\newcommand{\BCC}{BCC}

\newcommand{\resgraph}{skeleton\xspace}
\newcommand{\bcchead}{BCC head\xspace}
\newcommand{\fence}{fence}
\newcommand{\Fence}{Fence}
\newcommand{\aux}{auxiliary}

\newcommand{\plain}{plain}
\newcommand{\sketch}{skeleton}
\newcommand{\Sketch}{Skeleton}
\newcommand{\insketch}{\mf{InSkeleton}}
\newcommand{\sketches}{skeletons}
\newcommand{\sketchconnect}{skeleton-connectivity}
\newcommand{\connect}{connectivity}
\newcommand{\treepath}{\,\textasciitilde\,}

\newcommand{\stepfunc}[1]{\textit{#1}}
\newcommand{\firstcc}{\stepfunc{First-CC}}
\newcommand{\rootst}{\stepfunc{Rooting}}
\newcommand{\lowhigh}{\stepfunc{Tagging}}
\newcommand{\lastcc}{\stepfunc{Last-CC}}

\newcommand{\modelop}[1]{\texttt{#1}}
\newcommand{\forkins}{\modelop{fork}}

\newcommand{\thread}{thread}

\newcommand{\cas}{{\texttt{compare\_and\_swap}}}
\newcommand{\CAS}{{\texttt{CAS}}}

\mathchardef\sdash="2D

\newcommand{\true}{\emph{true}}
\newcommand{\false}{\emph{false}}

\newtheoremstyle{exampstyle}
{.05in} 
{.05in} 
{} 
{.5em} 
{\sc \bfseries} 
{.} 
{.5em} 
{} 
\theoremstyle{exampstyle} 
\theoremstyle{exampstyle} 

\makeatletter
\newenvironment{aproof}[1][\proofname]{\par
\pushQED{\qed}%
\normalfont
\topsep0pt \partopsep0pt 
\trivlist
\item[\hskip\labelsep
      \itshape
  #1\@addpunct{.}]\ignorespaces
}{%
\popQED\endtrivlist\@endpefalse
\addvspace{3pt plus 3pt} 
}

 \crefname{section}{Sec.}{Sec.}
 \crefname{theorem}{Thm.}{Thm.}
 \crefname{lemma}{Lem.}{Lem.}
 \crefname{corollary}{Col.}{Col.}
 \crefname{table}{Tab.}{Tab.}
 \crefname{algorithm}{Alg.}{Alg.}
 \crefname{figure}{Fig.}{Fig.}
 \crefname{fact}{Fact}{Fact}
\Crefname{table}{Tab.}{Tab.}

\begin{document}

\fancyhead{}

\title{Provably Fast and Space-Efficient Parallel Biconnectivity}
\settopmatter{authorsperrow=4}
  \author{Xiaojun Dong}
  \affiliation{\institution{UC Riverside}\city{}\country{}}
  \email{xdong038@ucr.edu}
  \author{Letong Wang}
  \affiliation{\institution{UC Riverside}\city{}\country{}}
  \email{lwang323@ucr.edu}
  \author{Yan Gu}
  \affiliation{\institution{UC Riverside}\city{}\country{}}
  \email{ygu@cs.ucr.edu}
  \author{Yihan Sun}
  \affiliation{\institution{UC Riverside}\city{}\country{}}
  \email{yihans@cs.ucr.edu}




\renewcommand\footnotetextcopyrightpermission[1]{} 

\begin{abstract}
Computing biconnected components (BCC) of a graph is a fundamental graph problem.
The canonical parallel BCC algorithm is the Tarjan-Vishkin algorithm, which
has $O(n+m)$ optimal work and polylogarithmic span on a graph with $n$ vertices and $m$ edges.
However, Tarjan-Vishkin is not widely used in practice.
We believe the reason is the space-inefficiency (it uses $O(m)$ extra space). 
In practice, existing parallel implementations are based on breath-first search (BFS).
Since BFS has span proportional to the diameter of the graph,
existing parallel BCC implementations suffer from poor performance on large-diameter graphs
and can be slower than the sequential algorithm on many real-world graphs.

We propose the first parallel biconnectivity algorithm (\ouralgo{}) that has optimal work, polylogarithmic span, and is space-efficient.
Our algorithm creates a \sketch{} graph based on any spanning tree of the input graph. 
Then we use the connectivity information of the \sketch{} to compute the biconnectivity of the original input.
We carefully analyze the correctness of our algorithm, which is highly non-trivial.

We implemented \ouralgo{} and compared it with existing implementations,
including \gbbs{}, Slota and Madduri's algorithm, and the sequential Hopcroft-Tarjan algorithm.
We tested them on a 96-core machine on 27 graphs with varying edge distributions.
\ouralgo{} is the fastest on \emph{all} graphs.
On average (geometric means), \ouralgo{} is 
3.1$\times$ faster than the \emph{best existing baseline} on each graph.
\end{abstract}




\hide{
\textbf{
{\LARGE This is the full version of the submission }}

\vspace{.5in}

\begin{center}
\Large
\textbf{\emph{Provably Fast and Space-Efficient Parallel Biconnectivity}}
\end{center}

\vspace{.5in}

{\LARGE
All the supplemental materials
are presented \textbf{in the appendix} (the main body is the same as our submission).
We keep the main body and references because our supplemental materials use hyperlinks and citations to the previous parts of the paper.}
}

\def\@titlefont{\huge\sffamily\bfseries} 

\maketitle

\section{Introduction}\label{sec:intro}
Graph biconnectivity is one of the most fundamental graph problems.
Given an \emph{undirected} graph $G=(V,E)$ with $n=|V|$ vertices and $m=|E|$ edges,
a \defn{connected component (CC)} is a maximal subset in $V$
such that every two vertices in it are connected by a path.
A \defn{biconnected component (BCC)} (or blocks) is a maximal subset $C\subseteq V$
such that $C$ is connected and remains connected after removing any vertex $v\in C$.
In this paper, we use BCC (or CC)
for both the biconnected (or connected) component in the graph
and the problem of computing all BCCs (or CCs).
BCC has extensive applications such as planarity testing~\cite{boyer2006simplified,hopcroft1974efficient,bachmaier2005radial}, centrality computation~\cite{sariyuce2013incremental,jamour2017parallel,sariyuce2013shattering}, and network analysis~\cite{newman2008bicomponents,ausiello2011real}.

Sequentially, the Hopcroft-Tarjan algorithm~\cite{hopcroft1973algorithm} 
for BCC uses $O(n+m)$ work.
However, this algorithm requires generating a spanning tree of~$G$ based on the depth-first search (DFS), which is considered hard to be parallelized~\cite{reif1985depth}.
Later, Tarjan and Vishkin proposed the canonical parallel BCC algorithm~\cite{tarjan1985efficient}.
It uses an \defn{arbitrary spanning tree (AST)} (a spanning tree with any possible shape) of the graph instead of the depth-first tree.
Tarjan-Vishkin algorithm has $O(n+m)$ optimal work (number of operations) and polylogarithmic span (longest dependent operations), assuming an efficient parallel CC algorithm.

Although the Tarjan-Vishkin algorithm is theoretically considered ``optimal'' in work and span, 
significant challenges still remain in achieving a high-performance implementation in practice. 
The main issue in Tarjan-Vishkin is space-inefficiency.
Tarjan-Vishkin generates an auxiliary graph $G'=(V',E')$ (which we refer to as the \defn{\sketch{}}), where every edge $e \in E$ maps to a vertex in $V'$.
Tarjan and Vishkin showed that computing CC on $G'$ gives the BCC on $G$, and we refer to this step as the \defn{connectivity} phase.
This \defn{\sketchconnect{}} framework is adopted in many later papers.
Such algorithms first generate a \sketch as an auxiliary graph $G'$ from $G$, and then finds the CCs on $G'$ that reflect BCC information on the input graph $G$.
%
Unfortunately, in Tarjan-Vishkin, generating the \sketch{} $G'$ and computing CC on $G'$ take $O(m)$ extra space, which greatly increases the memory usage and slows down the performance. 

\begin{figure}
  \centering
  \vspace{.5em}
  \includegraphics[width=\columnwidth]{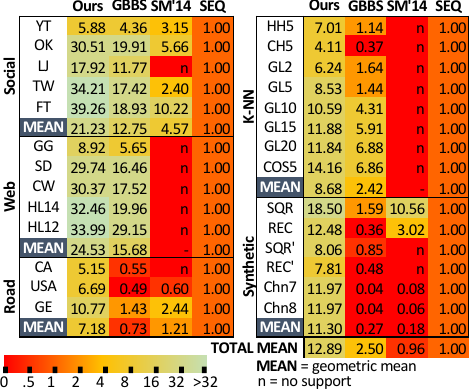}
  \caption{\small \textbf{The heatmap of relative speedup for parallel BCC algorithms over the sequential Hopcroft-Tarjan algorithm~\cite{hopcroft1973algorithm} using 96 cores (192 hyper-threads).}
Larger/green means better. The numbers indicate how many times a parallel algorithm is faster than sequential Hopcroft-Tarjan ($<1$ means slower). The two baseline algorithms are from~\cite{slota2014simple,gbbs2021}. 
Complete results are in~\cref{tab:bcc}.}\label{fig:heatmap}
\Description[<short description>]{<long description>}
\vspace{-.2em}
\end{figure}
%
In practice, most existing parallel BCC implementations also follow the \defn{\sketchconnect{}} framework but
overcome the space issue by using other \sketches{} based on breadth-first search (BFS) trees~\cite{cong2005experimental,slota2014simple,gbbs2021,chaitanya2015simple,chaitanya2016efficient,wadwekar2017fast,feng2018distributed}.
These algorithms either use \sketches{} with $O(n)$ size~\cite{cong2005experimental,chaitanya2015simple,chaitanya2016efficient,wadwekar2017fast,feng2018distributed}
or maintain implicit \sketches{} with $O(n)$ auxiliary space~\cite{slota2014simple,gbbs2021}.
We say a BCC algorithm is \defn{space-efficient} if it uses $O(n)$ auxiliary space (other than the input graph).
However, since computing BFS has span proportional to the graph,
these BFS-based algorithms can be fast on small-diameter graphs (e.g., social and web graphs),
but have poor performance on large-diameter graphs (e.g., $k$-nn and road graphs).
In our experiments, we observe that existing parallel implementations can even be slower than sequential Hopcroft-Tarjan on many real-world graphs (see \gbbs{}~\cite{gbbs2021} and \hipc{}~\cite{slota2014simple} in \cref{fig:heatmap}).

\defn{In this paper, we give the first space-efficient ($\,O(n)$ auxiliary space) parallel biconnectivity algorithm that has efficient $O(m+n)$ work and polylogarithmic span}.
Our \sketch{} $G'$ is based on an arbitrary spanning tree (AST).
Unlike Tarjan-Vishkin, our $G'$ is a subgraph of $G$ and can be maintained \emph{implicitly} in $O(n)$ auxiliary space.
The key idea is to carefully identify some \defn{\fence} edges, which indicate the ``boundaries'' of the BCCs.
At a high level, we categorize all graph edges into \fence{} tree edges, plain (non-\fence{}) tree edges, back edges, and cross edges.
Our \sketch{} $G'$ contains the plain tree edges and cross edges.
Using $O(n)$ space, we can efficiently determine the category of each edge in $G$.
When processing the \sketch{}, we use the input graph $G$ but skip the \fence{} and back edges.
We show that the BCC information of $G$ can be constructed from the CC information of $G'$ plus some simple postprocessing.
Since our algorithm is based on \textbf{Fencing an Arbitrary Spanning Tree},
we call our algorithm \textbf{\ouralgo{}}.
More details of \ouralgo{} are in \cref{fig:bcc}.
We note that conceptually our algorithm is simple, but the correctness analysis is highly non-trivial.

We implement our theoretically-efficient \ouralgo{} algorithm
and compare it to the state-of-the-art parallel BFS-based BCC implementations \gbbs{}~\cite{gbbs2021} and \hipc~\cite{slota2014simple}.
We also compare \ouralgo{} to the sequential Hopcroft-Tarjan algorithm.
We test 27 graphs, including social, web, road, \knn{}, and synthetic graphs, with significantly varying sizes
and edge distributions. The details of the graphs and results are given in \cref{tab:bcc}.
We also show the relative running time in \cref{fig:heatmap}, normalized to the sequential Hopcroft-Tarjan.

On a machine with 96 cores, \ouralgo{} is the fastest on \defn{all} tested graphs.
We use the geometric means to compare the ``average'' performance across multiple graphs.
Due to work- and space-efficiency, our algorithm running on one core is competitive with Hopcroft-Tarjan (2.8$\times$ slower on average).
Polylogarithmic span leads to good parallelism for \defn{all types of graphs} (15--66$\times$ self-relative speedup on average).
On small-diameter graphs (social and web graphs), although \gbbs{} and \hipc{} also achieve good parallelism,
\ouralgo{} is still 1.2--2.1$\times$ faster than the best of the two,
and is 5.9--39$\times$ faster than sequential Hopcroft-Tarjan.
For large-diameter graphs (road, $k$-nn, grid, and chain graphs), existing BFS-based implementations can perform worse than Hopcroft-Tarjan.
Due to low span,
\ouralgo{} is 1.7--295$\times$ faster than \gbbs{} (10$\times$ on average), and 4.1--18.5$\times$ faster than sequential Hopcroft-Tarjan (9.2$\times$ on average).
On all graphs, \ouralgo{} is 3.1$\times$ faster on average than the best of the three existing implementations.
Our code is publicly available~\cite{bcccode}. We present more results and analyses in the full version of this paper~\cite{dong2023provablyfull}.


\hide{
We summarize our contribution as follows:
\begin{enumerate}
    \item The design and analysis of a parallel biconnectivity algorithm that is theoretically-efficient in work, span, and space.
    \item A practical implementation on the parallel biconnectivity problem, which achieves the best performance on a wide range of real-world and synthetic graphs. This is the first publicly available theoretically-efficient implementation to our best knowledge.
    \item In-depth analysis on existing parallel biconnectivity implementations.
\end{enumerate}
}

\section{Preliminaries}\label{sec:prelim}

\myparagraph{Computational Model}.
We use the work-span (or work-depth) model for fork-join parallelism
with binary forking to analyze parallel algorithms~\cite{CLRS,blelloch2020optimal},
which is recently used in many papers on parallel algorithms~\cite{agrawal2014batching,blelloch2010low,BCGRCK08,BG04,Blelloch1998,blelloch1999pipelining,BlellochFiGi11,BST12,BBFGGMS16,dinh2016extending,xu2022efficient,dong2021efficient,blelloch2018geometry,dhulipala2020semi,BBFGGMS18,blelloch2020randomized,gu2021parallel,gu2022parallel,shen2022many}.
We assume a set of \thread{}s that share a common memory.
A process can \forkins{} two child software \thread{s} to work in parallel.
When both children complete, the parent process continues.
The \defn{work} of an algorithm is the total number of instructions and
the \defn{span} (depth) is the length of the longest sequence of dependent instructions in the computation.
We say an algorithm is \defn{work-efficient} if its work is asymptotically the same as the best sequential algorithm.
We can execute the computation using a randomized work-stealing scheduler~\cite{BL98,ABP01} in practice.
We assume unit-cost atomic operation \cas{}$(p,v_{\mathit{old}},v_{\mathit{new}})$ (or \CAS{}),
which atomically reads the memory location pointed to by $p$, and write value $v_{\mathit{new}}$ to it if the current value is $v_{\mathit{old}}$.
It returns $\true{}$ if successful and $\false{}$ otherwise.

\myparagraph{Notation.}
Given an undirected graph $G=(V,E)$,
we use $n=|V|$, $m=|E|$.
Let $\textsf{diam}(G)$ be the diameter of $G$, and $x$--$y$ be an edge between $x$ and $y$.
\defn{CC} and \defn{BCC} are defined in \cref{sec:intro}.
An \defn{articulation point} (or cut vertex) is a vertex s.t.\ removing it increases the number of CCs.
A \defn{bridge} (or cut edge) is an edge s.t.\ removing it increases the number of CCs.
A spanning tree $T$ of a connected graph $G$ is a spanning subgraph of $G$
that contains no cycles.
The spanning forest is defined similarly if $G$ is disconnected.
For simplicity, we assume $G$ is connected,
but our algorithm and implementation work on any graph.
Given a graph $G$ and a rooted spanning tree {$T$},
an edge is a \defn{tree edge} if it is in $T$.
A non-tree edge is a \defn{back edge} if one endpoint is the ancestor of the other endpoint,
and a \defn{cross edge} otherwise.
\cref{fig:bcc} Step 3 shows an illustration.
If $T$ is a BFS tree, there are no back edges; if $T$ is a DFS tree, there are no cross edges.
We use $x\,$\textasciitilde$\,y$ to denote the tree path between $x$ and $y$ on $T$.
We denote the parent of vertex $u$ as $p(u)$, and the subtree of $u$ as $T_u$.
The notation used in this paper is given in \cref{tab:notation}.

We use $O(f(n))$ \emph{with high probability} (\whp{}) in $n$ to mean $O(cf(n))$ with probability at least $1-n^{-c}$ for $c \geq 1$.

\myparagraph{Euler tour technique (ETT).}
ETT is proposed by Tarjan and Vishkin~\cite{tarjan1985efficient} in their BCC algorithm to root a spanning tree.
Later, ETT becomes a widely-used primitive in both sequential and parallel settings, including computational geometry~\cite{aggarwal1988parallel}, graph algorithms~\cite{chiang1995external,vishkin1985efficient,arge2002cache}, maintaining subtree or tree path sums~\cite{CLRS}, and many others.
ETT is needed in Tarjan-Vishkin because when an arbitrary spanning tree is generated for a graph (e.g., from a CC algorithm), it is not
rooted, and thus we do not have the parent-child information for the vertices.
Given an unrooted tree $T$ with $n-1$ edges, ETT finds an Euler tour of $T$, which is a cycle traversing each edge in $T$ exactly twice (once in each direction).
ETT first constructs a linked list on the $2n-2$ directed tree edges, and runs list ranking on it.
We refer the audience to the textbooks on parallel algorithms~\cite{JaJa92,Reif93} for more details on ETT.
Using the semisort algorithm from~\cite{gu2015top,blelloch2020optimal} and list ranking from~\cite{blelloch2020optimal}, ETT costs $O(n)$ expected work and $O(\log n)$ span \whp{}.
Given $T$, we can set any vertex as the root of $T$, and use ETT to determine the directions of the edges.
We can then determine the parent of any vertex, and whether an edge is a tree edge, back edge, or cross edge in $O(1)$ work.


\begin{table}
\centering \small
\begin{tabular}{@{}c@{}@{}l@{ }@{ }c@{}l@{}}
  \hline
  $G=(V,E)$ & : Input Graph &
  $T=(V,E_T)$ & : A spanning tree in $G$
  \\
  \multicolumn{3}{c}{$a,b,c,u,v,h,w,x,y,z,u',v',c'\dots \in V$} & : Vertices in $G$\\\hline\hline
  $x$--$y\in E$   & : An edge in $G$ &
  $C,C_i$ & : A BCC in $G$ \\
  $T_u$  & : $u$'s subtree in $T$ &
  $h_C$ & : The BCC head of $C$ \\
  $p(u)$ & : $u$'s parent in $T$ &
   $x\,$\treepath$\,y$   & : A tree path in $T$\\
   $P=x$--$y$--$\cdots$&: A path & $G'$ & : The {\sketch{}}\\
   \hline
   \hline
  \textbf{Fence edge}&\multicolumn{3}{@{}l@{}}{: $(p(v),v)\in E_T$, $\nexists$ $(x,y)\in E$, s.t. $x\in T_v$ and $y\notin T_{p(v)}$}\\
  &\multicolumn{3}{@{}l@{}}{\hspace{-.75em}(no edge from $v$'s subtree escapes from $p(v)$'s subtree)}\\
  \textbf{Plain edge}&\multicolumn{3}{@{}l@{}}{: $(p(v),v)\in E_T$, $(p(v),v)$ is not a fence edge}\\
  \multicolumn{4}{@{ }l}{\textbf{Back edge, Cross edge}~: Edges in $E\setminus E_T$, defined as usual}\\
  \multicolumn{4}{@{ }l@{}}{\textbf{\Sketch{}} $G'=(V,E')$ in \textbf{\ouralgo{}}\,: $E'=\{$plain$\,\&\,$cross edges$\}$}\\
  \hline
\end{tabular}
\vspace{-0.8em}
\caption{\bf Notations and terminologies in this paper.} 
\vspace{-.5em}
\label{tab:notation}
\end{table}

\hide{
\myparagraph{Parallel Primitives.}
\defn{Reduce} takes as input a sequence [$a_1, a_2, \cdots, a_n$], a binary associative operator $\oplus$, and outputs $\oplus_{i=1}^n a_i$
(e.g., if $\oplus$ is $+$, it outputs the total sum of the sequence plus $t$).
\defn{Scan} takes the same input as reduce, but rather than only return the total sum,
it returns a sequence that consists of the running sum
$[a_1, a_1 \oplus a_2, \cdots, \oplus_{i=1}^n a_i]$.
\defn{Pack} takes an array of elements and a binary array of the same size,
and returns an array containing elements for which its corresponding value in the binary array is true.
\defn{Semisort} takes as input an array of records with associated keys, and return a reordered array
such that records with identical keys are contiguous.
Note that the problem does not require all keys appear in sorted order in the output, nor the records with the same key sorted.
\defn{List ranking} takes a linked list with associated value on each node, and returns the rank of each node in the linked list.
One can flatten the linked list to a contiguous array such that it can be random access.
The \defn{Euler tour} of a spanning tree $T$ is a directed circuit such that
it traverses each edge on both directions exactly once.

Given a spanning tree $T$, the \defn{Euler tour} of $T$ is a directed circuit such that it traverses each edge exactly once.
This technique allows the parallel computations on various kinds of information of a tree in linear time and polylogarithmic span.
Most computations can be represented by two operations, rootfix and leaffix.
The \defn{rootfix} of a vertex $v$ is the aggregated value on the subtree of $v$ (e.g., the subtree size of $v$).
The \defn{leaffix} of a vertex $v$ is the aggregated value on all vertices from $v$ to the root of the tree (e.g., the depth of $v$).
Specifically, in the Tarjan-Vishkin algorithm, this technique helps to compute the preorder number, postorder number, and the number of descendants of each vertex.

} 

\setlength{\abovedisplayskip}{-.1in}
\setlength{\belowdisplayskip}{.05in}
\setlength{\abovedisplayshortskip}{-.1in}
\setlength{\belowdisplayshortskip}{-.1in}

\section{Existing BCC Algorithms} \label{sec:previous-bcc}

This section reviews the existing BCC algorithms and implementations.
We will use the \emph{\sketchconnect{} framework} to describe the existing BCC algorithms.
The \emph{skeleton phase} generates a \sketch{} $G'$ from $G$, which is an auxiliary graph. 
Then the \emph{connectivity phase} computes the connectivity on $G'$ to construct the BCCs of $G$.
Existing BCC algorithms can be categorized by how the \sketch{} $G'$ is generated.
The Hopcroft-Tarjan algorithm uses DFS-based \sketches{}; the Tarjan-Vishkin Algorithm generates a \sketch{} based on an arbitrary spanning tree (AST); almost all other BCC algorithms (see \cref{sec:related-other}) use BFS-based \sketches{}.


\subsection{The Hopcroft-Tarjan Algorithm}\label{sec:hopcroft-tarjan}
Sequentially, Hopcroft-Tarjan BCC algorithm~\cite{hopcroft1973algorithm} has $O(n+m)$ work using a depth-first search (DFS) tree $T$.
Based on $T$, two \defn{tags} $\first[\cdot]$ and $\low[\cdot]$ are assigned to each vertex.
$\first[v]$ is the preorder number of each vertex in $T$. 
$\low[v]$ gives the earliest (smallest preorder) vertex incident on any vertex $u\in T_v$ via a non-tree edge and $u$ itself.
More formally,

\begin{align*}
\low[v] &= \min\{w_1[u]~|~u\in V \text{ is in the subtree rooted at }v\}\\[-.05in]
w_1[u]&=\min\{\{\first[u]\} \cup \{\first [u']~|~(u, u')\notin T\}\}
\end{align*}

To compute the BCCs, an additional stack is maintained.
Each time we visit a new edge, it is pushed into the stack.
When an articulation point $p(u)$ is found by $u$ ($\low[u] \geq \first[p(u)]$), edges are popped from the stack until $u$--$p(u)$ is popped.
These edges and the relevant vertices form a BCC.

Conceptually, the \sketch{} in Hopcroft-Tarjan is the DFS tree without the ``fence edges'' of $u$--$p(u)$ when $\low[u] \geq \first[p(u)]$.
This insight also inspires our BCC algorithm.

\subsection{The Tarjan-Vishkin Algorithm}\label{sec:tarjan-vishkin}
Hopcroft-Tarjan uses a DFS tree as the \sketch, but DFS is inherently serial and hard to be parallelized~\cite{reif1985depth}.
To parallelize BCC, the Tarjan-Vishkin algorithm~\cite{tarjan1985efficient} uses an arbitrary spanning tree (AST) instead of a DFS tree.
This spanning tree $T$ can be obtained by any parallel CC algorithm.
The algorithm then uses ETT (which was also proposed in that paper) to root the tree $T$ (see \cref{sec:prelim}).
Then the algorithm builds a \sketch{} $G'=(E,E')$ and runs a connectivity algorithm on it.
\ifconference{We describe $G'$ in more details in the full paper~\cite{dong2023provablyfull},}\iffullversion{
We describe $G'$ in more details in \cref{app:tv},}
and only briefly review it here.
The vertices in $G'$ correspond to the edges in $G$\footnote{In a later paper~\cite{edwards2012better}, it was shown that the number of vertices in $G'$ can be reduced to $O(n)$, but $|E'|$ is still $O(m)$.}.
To determine the edges in $G'$, the algorithm uses four tags ($\first[\cdot]$, $\last[\cdot]$, $\low[\cdot]$, and $\high[\cdot]$) for each vertex.
Here $\first[u]$ and $\last[u]$ are the first and last appearance of vertex $u$ in the Euler tour (note that this is not the same $\first[\cdot]$ in Hopcroft-Tarjan, but conceptually equivalent).
$\low[\cdot]$ is the same as defined in Hopcroft-Tarjan, and $\high[\cdot]$ is defined symmetrically:

\begin{align*}
\high[v] &= \max\{w_2[u]~|~u\in V \text{ is in the subtree rooted at }v\}\\[-.05in]
w_2[u]&=\max\{\{\first[u]\} \cup \{\first [u']~|~(u, u')\notin T\}\}
\end{align*}

All tags can be computed in $O(n+m)$ expected work and $O(\log n)$ span \whp{} using ETT.
Tarjan-Vishkin then finds the CCs on $G'$ to compute the BCCs of $G$.
However, $G'$ in Tarjan-Vishkin can be large, making the algorithm less practical.

Assuming an efficient ETT and a parallel CC algorithm, Tarjan-Vishkin uses $O(n+m)$ optimal expected work and polylogarithmic span.
However, the space-inefficiency hampers the practicability of Tarjan-Vishkin since $G'$ contains $O(m)$ edges.
In our experiments, Tarjan-Vishkin takes up to 11$\times$ extra space than our \ouralgo or \gbbs{}.
On our machine with 1.5TB memory,
Tarjan-Vishkin ran out of memory when processing the Clueweb graph~\cite{webgraph},
although it only takes about 300GB to store the graph \ifconference{(see discussions in the full version~\cite{dong2023provablyfull}).}\iffullversion{(see discussions in \cref{sec:exp:tv}).}
The large space usage forbids running Tarjan-Vishkin on large-scale graphs on most multicore machines.
Even for small graphs, high space usage can increase memory footprint and slow down the performance.


Some existing BCC implementations (e.g., \gbbs{}~\cite{gbbs2021} and TV-filter~\cite{cong2005experimental})
were also described as Tarjan-Vishkin algorithms,
probably because they also use the \sketchconnect{} framework.
We note that their correctness relies on BFS-based \sketches{} (i.e., sparse certificates~\cite{cheriyan1991algorithms}), and we categorized them below together with a few other algorithms.

\subsection{Other Existing Algorithms / Implementations}\label{sec:related-other}

Before Tarjan-Vishkin, Savage and J{\'a}J{\'a}~\cite{savage1981fast} showed a parallel BCC algorithm based on matrix-multiplication with $O(n^3\log n)$ work.
Tsin and Chin~\cite{tsin1984efficient} gave an algorithm that uses an AST-based skeleton.
It is quite similar to Tarjan-Vishkin, but uses $O(n^2)$ work. 
There are other parallel BCC algorithms based on ear decomposition~\cite{miller1987new,ramachandran1992parallel,savage1981fast,Reif93}.
Similar to Tarjan-Vishkin, a faithful implementation would take $O(m)$ auxiliary space. 

To achieve space-efficiency, many later parallel BCC algorithms use BFS-based \sketches{}~\cite{cong2005experimental,slota2014simple,gbbs2021,chaitanya2015simple,chaitanya2016efficient,wadwekar2017fast,feng2018distributed,ji2020aquila}.
Many of them use the similar idea of sparse certificates~\cite{cheriyan1991algorithms}.
BCC is much simpler with a BFS tree---all non-tree edges are cross edges with both endpoints in the same or adjacent levels.
Cong and Bader's TV-filter algorithm~\cite{cong2005experimental} uses the \sketch{} as the BFS tree $T$ and an arbitrary spanning tree/forest for $G\setminus T$ ($O(n)$ total size).
Slota and Madduri's algorithms~\cite{slota2014simple} and Dhulipala et al.'s algorithm~\cite{gbbs2021} use the \sketches{} as the input graph $G$ excluding $O(n)$ vertices/edges.
The other algorithms~\cite{chaitanya2015simple,chaitanya2016efficient,wadwekar2017fast,feng2018distributed} use a BFS tree as the \sketch{}, and compute connectivity dynamically.
All these algorithms are space-efficient.
Their \sketch{} graphs either have $O(n)$ size~\cite{cong2005experimental,chaitanya2015simple,chaitanya2016efficient,wadwekar2017fast,feng2018distributed} or can be implicitly represented using $O(n)$ information~\cite{slota2014simple,gbbs2021}.
However, the span to generate a BFS tree is proportional to the diameter of the graph, which is inefficient for large-diameter graphs.
\hide{
In summary, while these algorithms are different, they share three following common properties.
(1) They all use the special property for BFS trees (e.g., no back edge).
(2) Their sketches either have $O(n)$ size~\cite{cong2005experimental,chaitanya2015simple,chaitanya2016efficient,wadwekar2017fast,feng2018distributed} or can be implicitly represented using $O(n)$ information~\cite{slota2014simple,gbbs2021}.
(3) Generating the BFS tree requires $\Omega(\textsf{diam}(G))$ span, which is slow for graphs with large diameters.
}

\hide{
The Tarjan-Vishkin algorithm then reduces the the BCC problem to a CC problem on a transformed graph $G'=(V', E')$, where $V'=E\setminus T$ and $E'$ can be computed based on identifying tree edges, back edges, and cross edges in $E$.
The Tarjan-Vishkin costs $O(W_{CC}(n, m)))$ work and $O(D_{CC}(n, m))$ span,
where $W_{CC}(n, m)$ and $D_{CC}(n, m)$ are the work and span for graph connectivity.
It can be done in $O(n+m)$ work and polylogarithmic span using the linear-work parallel connectivity algorithm~\cite{SDB14} proposed later.
We refer the readers to~\cite{tarjan1985efficient} for more details about the Euler tour technique and Tarjan-Vishkin algorithm.

The Tarjan-Vishkin algorithm is theoretically-efficient.
However, we are unaware of any existing publicly-available implementation of it. We believe it is mainly due to the space-inefficiency in processing the new graph $G'$.
Although running connectivity requires space proportional to the number of vertices,
the number of vertices in $G'$ is $|E|$, which is the number of edges in $G$.
Most real-world graphs have significantly more edges than vertices (see \cref{tab:graphinfo}),
which means it can increase tremendous auxiliary space.
For instance, saving all edges in Hyperlink12 uses 842GB of memory,
while saving all vertices in the same graph only uses 13.2GB of memory.
This can greatly increase memory footprint and even make $G'$ too large to fit in memory.
This space overhead makes the Tarjan-Vishkin algorithm less suitable in the shared-memory parallel setting.

Sequentially, Hopcroft-Tarjan algorithm~\cite{hopcroft1973algorithm} finds all \BCC{}s in $O(m+n)$ work using a DFS tree.
The idea of Hopcroft-Tarjan algorithm is to identify the back edges using the label $\low[v]$ for $v\in V$:

and $\first[u]$ is the preorder of $u\in V$ based on the DFS tree.
The BCCs can be computed once $\first[u]$ and $\low[u]$ are known.

However, DFS is generally considered to be hard to parallelize~\cite{reif1985depth}.
To achieve parallelism in BCC, the Tarjan-Vishkin algorithm~\cite{tarjan1985efficient}
uses a spanning tree $T$ generated by \emph{any} connectivity algorithm.
The given spanning tree can be in arbitrary shape and not necessarily be a DFS tree.
They then root the tree $T$ using the \emph{Euler tour technique} (ETT)~\cite{tarjan1985efficient},
which can also give $T$'s preorder and postorder.
Accordingly, the Tarjan-Vishkin algorithm needs to identify and deal with both the back edges and the cross edges.
In this case, in addition to $\low[\cdot]$, it also computes $\high[\cdot]$ as:
$$\high[v] = \max\{w_2[u]~|~u\in V \text{ is in the subtree rooted at }v\}$$
where $$w_2[u]=\max\{\first[u] \cup \{\first [u']~|~(u, u')\text{ is a non-tree edge}\}\}$$
Then we can identify back edges and cross edges using $\low[\cdot]$ and $\high[\cdot]$, as well as the preorder of $T$ ($\first[\cdot]$) and the postorder of $T$ ($\last[\cdot]$),
which all can be obtained by the Euler tour.

To address the space issue, many existing parallel BCC implementations use BFS trees~\cite{gbbs2021, cong2005experimental, slota2014simple}
to avoid processing the graph $G'$. This also simplifies algorithm design since using BFS trees leads to no back edges.
However, using BFS trees sacrifices the theoretical efficiency---the span bound becomes $\Omega(\textsf{diam}(G))$,
so these algorithms can be slower than the sequential algorithm on large-diameter graphs (see \cref{tab:bcc}).
To date, we are unaware of any BCC implementations that are theoretically-efficient.

The idea of the Hopcroft-Tarjan algorithm is to find the articulation points of the graph by using $\first$ and $\low$.
If $v$ (except the root) is an ancestor of vertex $u$ and $\low[u] \geq \first[v]$, $v$ is an articulation point.
If $v$ is the root and it has at least two different children $u_1, u_2$ satisfying the above condition, $v$ is also an articulation point.

Note that the information can only identify the articulation points and is insufficient to decompose a graph into BCCs.
To determine the BCCs of the graph, we need to maintain an additional stack.
Each time we visit a new edge, it is pushed into the stack.
When an articulation point $v$ is found by $u$, edges are popped from the stack until $(u, v)$ is moved. These edges form a BCC. The vertices incident on these edges can be labeled.

The Tarjan-Vishkin algorithm reduces the BCC problem to a CC problem on a transformed graph.
This was independently discovered by Tsin and Chin~\cite{tsin1984efficient}.
The Tsin-Chin algorithm runs in $O(n^2)$ work and $O(\log^2n)$ span, which is optimal on dense graphs but not ideal on sparse graphs.
Tarjan and Vishkin presented an algorithm that requires $O(n+m)$ work and polylogarithmic span using the linear-work parallel connectivity algorithm~\cite{SDB14} proposed later (which was the bottleneck in the original paper).
The main techniques of the Tarjan-Vishkin algorithm are in two aspects.

\myparagraph{Finding blocks.}
Let $R$ be the relation on the edges of $G$ defined by $e_1Re_2$ iff. $e_1=e_2$ or $e_1$ and $e_2$ are on a simple cycle of $G$.
The subgraphs of $G$ induced by the equivalence classes of $R$ are the BCCs.
An auxiliary graph $G'$ of $G$ is defined such that the CCs on $G'$ correspond to the BCCs of $G$.
The vertices of $G'$ are the edges of $G$.
Let $T$ be any spanning tree of $G$.
We denote the edges of $T$ by $u \rightarrow v$, where $u$ is the parent of $v$, denoted by $p(v)$.

Let the vertices of $T$ be numbered from 1 to $n$ in preorder.

}

\hide{
Most biconnectivity implementations are bounded by I/O, thus increasing memory footprints can impede practical performance.
An empirical solution to overcome the space overhead problem in the Tarjan-Vishkin algorithm is to use a BFS tree instead of an arbitrary spanning tree.
BFS trees do not have back edges, which simplifies the relationships and thus the auxiliary graph.
However, this breaks the theoretical guarantees in the Tarjan-Vishkin algorithm because the span of a BFS algorithm is proportional to the diameter of the graphs.
Therefore, these implementations perform favorably on small-diameter graphs but do not work well on large-diameter graphs (See~\cref{tab:bcc} for details).

\myparagraph{The Cong-Bader Algorithm.}
Cong and Bader implemented three algorithms, \textit{TV-SMP}, \textit{TV-opt}, and \textit{TV-filter} in~\cite{cong2005experimental}.
\textit{TV-SMP} is a straightforward emulation of the Tarjan-Vishkin algorithm. \textit{TV-opt} is an optimized version of \textit{TV-SMP} by reducing the number of parallel primitives and rearranging and merging some steps.
\textit{TV-filter} generates a BFS tree (rather than any spanning tree) such that some edges can be filtered out when building the auxiliary graphs.
Their results showed that \textit{TV-filter} outperformed the other two implementations by reducing the execution time in computing $\low$ and $\high$ values, labeling, and computing CCs.

\myparagraph{The Slota-Madduri Algorithm.}
Slota and Madduri implemented \textit{BiCC-BFS} and \textit{BiCC-Coloring} in~\cite{slota2014simple}.
Both algorithms are based on BFS and do not construct an auxiliary graph.
Although they have good practical performance, none of them have theoretical guarantees.
The upper bounds of work for these two algorithms are $O(nm)$ and $O(n^2)$, respectively.

\myparagraph{The \gbbs{} Implmentation.}
The graph based benchmark suite (\gbbs{})~\cite{gbbs2021} is a library that implements over 20 parallel graph algorithms, and biconnectivity is one of them.
The \gbbs{} implementation is also based on BFS.
It adapts the BC-labeling~\cite{BBFGGMS18} idea, such that the auxiliary graph is a subgraph of the original graph,
and thus the subsequent connectivity step can use only $O(n)$ extra space.

\subsection{The BC-labeling Technique}
Ben-David et al. proposed the \defn{BC-labeling} technique in~\cite{BBFGGMS18}.
Instead of labeling all $m$ edges as in the standard BCC output, BC-labeling assigns at most $2n$ total labels to \emph{vertices}.
In BC-labeling, each vertex $v$ (except for the root of the spanning tree) has a label indicating which BCC $v$ is in.
All vertices with the same label, plus another vertex called the \defn{component head} attached to this label, form a BCC.
See \cref{fig:bcc} as an example of BC-labeling.

Ben-David et al. presented the BC-labeling technique to represent BCCs in $O(n)$ space, but they did not formalize a BCC algorithm with the technique or prove its correctness.
Moreover, a glitch is found in their description. \xiaojun{Give more details or do not mention at all.}
}

\subsection{Space-Efficient BCC Representation}
Since some vertices (articulation points) appear in multiple BCCs (see \cref{fig:bcc} as an example), we need a representation of all BCCs in a space-efficient manner ($O(n)$ space).
We use a commonly used representation~\cite{BBFGGMS18,gbbs2021,feng2018distributed} in our algorithm.
Given a spanning tree $T$,
we assign a label for each vertex except for the root of $T$, indicating which BCC this vertex is in.
For all vertices with the same label, we find another vertex called the \defn{component head} (see details in \cref{sec:bcc-details}) attached to this label. All vertices with the same label and the corresponding component head form a BCC.
An example of this representation is given in \cref{fig:bcc}.
It is easy to see that this representation uses $O(n)$ space since we have $n-1$ labels for all vertices and at most $n-1$ component heads. 

\newcommand{\nosemic}{\renewcommand{\@endalgocfline}{\relax}}
\newcommand{\dosemic}{\renewcommand{\@endalgocfline}{\algocf@endline}}
\newcommand{\popline}{\Indm\dosemic}
\newcommand{\pushline}{\Indp}
\setlength{\algomargin}{.5em}
\begin{algorithm}[t]
\SetNoFillComment
\small
\caption{The \ouralgo algorithm\label{alg:bcc}}
\KwIn{An undirected graph $G=(V,E)$}
\KwOut{The labels $l[\cdot]$ for vertices, and the component head for each BCC}
\SetKwInOut{Maintains}{Maintains}
\SetKwProg{myfunc}{Function}{}{}
\SetKwFor{parForEach}{ParallelForEach}{}{endfor}
\DontPrintSemicolon

Compute the spanning forest $F$ of $G$ \label{line:firstcc}\Comment{First CC}\\
Root all trees in $F$ using the Euler tour technique  \label{line:ett}\Comment{Rooting}\\
Compute tags (e.g., $\low$, $\high$) of each vertex based on the Euler tour\label{line:low-high} \Comment{Tagging}\\
Compute the vertex label $l[\cdot]$ using connectivity on $G$ with edges satisfying \insketch$(u, v)=\true$\label{line:lastcc}\Comment{Last CC}\\
\parForEach {$u\in V$ with $l[u]\ne l[p(u)]$\label{line:component-head0}} {
  Set the component head of $l[u]$ as $p(u)$\label{line:component-head}
}
\medskip
\myfunc{\upshape \insketch{}$(u, v)$\label{line:bcc-cond}\Comment{Decide if $u$--$v$ is in \sketch{} $G'$}} {
  \If{$(u, v)$ is a tree edge\label{line:bcc-tree}} {
    \Return{$\neg$ \mf{\Fence}($u, v$) and $\neg$ \mf{\Fence}($v, u$)\label{line:bcc-fence}}
  } \lElse {
    \Return{$\neg$ \mf{Back}($u, v$) and $\neg$ \mf{Back}($v, u$)\label{line:bcc-back}}
  }
}
\myfunc{\upshape \mf{\Fence}$(u, v)$\Comment{Decide if tree edge is fence edge}} {
  \Return{$\first[u] \leq \low[v]$ and $\last[u] \geq \high[v]$\label{line:bcc-critical}}
}

\myfunc{\upshape \mf{Back}$(u, v)$\Comment{Decide if non-tree edge is back edge}} {
  \Return{$\first[u] \leq \first[v]$ and $\last[u] \geq \first[v]$\label{line:bcc-backward}}
}
\end{algorithm} 

\section{The \Ouralgo{} Algorithm}

In this section, we present our \ouralgo{} algorithm with analysis.
Our algorithm is the first parallel BCC algorithm that is work-efficient, space-efficient, and has polylogarithmic span.
Recall that BFS-based algorithms are space-efficient, but BFS itself does not parallelize well.
Tarjan-Vishkin is based on AST and is highly parallel, but generating the \sketch{} is space-inefficient.
To achieve both high parallelism and space efficiency, we need novel algorithmic insights.

Interestingly, our key idea is to revisit the sequential DFS-based Hopcroft-Tarjan algorithm (\cref{sec:hopcroft-tarjan}).
Although DFS is inherently sequential, the insights in Hopcroft-Tarjan inspire our parallel BCC algorithm.
The (implicit) \sketch{} in Hopcroft-Tarjan is simple and the \sketch{} size is small ($O(n)$).
Unlike many later parallel BCC algorithms with the high-level ideas to combine cycles (based on \cref{lem:cycle}),
the idea in Hopcroft-Tarjan is the ``fencing'' condition as follows.
When computing the CC on the \sketch{} $G'$ (the DFS tree) and traversing the edge from $v$ to $p(v)$, the CC on $G'$ (BCC on $G$) is \emph{fenced} if $\low[v] \geq \first[p(v)]$.
This condition partitions the DFS tree $T$ into multiple CCs that correspond to BCCs in $G$.
Note that $G'$ in Hopcroft-Tarjan only contains edges from the DFS tree, because there are no cross edges in DFS trees and
all back edge information is captured by $\low[\cdot]$.

Now we try to generalize this idea to an arbitrary spanning tree (AST).
Directly using the ``fencing'' condition in Hopcroft-Tarjan does not work since we need to deal with cross edges. 
Note that a fence edge $v$--$p(v)$ in Hopcroft-Tarjan means that
\emph{vertices in $u$'s subtree do not have an edge that escapes (i.e., the other endpoint is outside) $p(u)$'s subtree}.
We define our \fence{} edges also based on this condition.
More formally, we say a tree edge $(u,v)$ where $u=p(v)$ is a \fence{} edge
if there is no edge $(x,y)\in E$ such that $x\in T_v$ and $y\notin T_u$.
Intuitively, it means $v$'s subtree $T_v$ is ``isolated'' from other parts outside $p(v)$'s subtree, and only interacts with the outside world through $p(v)$.
To get an equivalent condition for an AST, we borrow the idea from Tarjan-Vishkin and also compute four axillary arrays $\first[\cdot]$, $\last[\cdot]$, $\low[\cdot]$, and $\high[\cdot]$.
The ``fencing'' condition then becomes $\low[v] \geq \first[p(v)]$ and $\high[v] \leq \last[p(v)]$.
A non-fence tree edge is referred to as a \defn{plain edge}.
Note that the information for back edges is already captured by the $\low[\cdot]$ and $\high[\cdot]$ arrays, which will also be used to decide fence edges.
Our algorithm will ignore back edges as in Hopcroft-Tarjan, and our \sketch{}~$G'$ contains plain tree edges and cross edges.
Since the main approach in our algorithm is Fencing an Arbitrary Spanning Tree, we call our algorithm \textbf{\ouralgo{}}.
We note that the high-level idea of fencing (finding some special edges on the spanning tree) is also used in some existing work~\cite{BBFGGMS18,gbbs2021,slota2014simple}.
Our design of the \sketch{} and the fencing condition is
the first to achieve work-efficiency, polylogarithmic span, and space-efficiency for the BCC problem.

The outline of the algorithm is given in \cref{fig:bcc}, and the pseudocode is in \cref{alg:bcc}.
Although our fencing algorithm is simple, we note that formally proving the correctness (\cref{sec:correctness}) is highly non-trivial.


\subsection{Algorithmic Details}\label{sec:bcc-details}

Our \ouralgo{} algorithm has four steps: \firstcc{} (generate spanning trees), \rootst{} (root the spanning trees using ETT), \lowhigh{} (compute $\first[\cdot]$, $\last[\cdot]$, $w_1[\cdot]$, $w_2[\cdot]$, $\low[\cdot]$, $\high[\cdot]$, $p[\cdot]$), and \lastcc{} (run CC on the \sketch{} and post-processing).
In the \sketchconnect{} framework, the first three steps are the \sketch{} phase (compute the \sketch{} $G'$),
and the last step is the \connect{} phase (run CC on $G'$ to find all BCCs in $G$).

\myparagraph{\firstcc{}} (Step 1 in \cref{fig:bcc}, \cref{line:firstcc} in \cref{alg:bcc}).
This step finds all CCs in $G$ and generates a spanning forest $F$ of $G$.
For simplicity, in the following, we focus on one CC and its spanning tree $T$, which is unrooted at this moment.
If $G$ contains multiple CCs, they are simply processed in parallel.
Running CC only requires $O(n)$ \aux{} space.

\myparagraph{\rootst{}} (Step 2 in \cref{fig:bcc}, \cref{line:ett} in \cref{alg:bcc}).
We use the Euler tour technique (ETT) in \cref{sec:prelim} to root $T$, which implies the tree edge directions (\cref{fig:bcc}, Step 2).
ETT requires $O(n)$ space.

\myparagraph{\lowhigh{}} (Step 3 in \cref{fig:bcc}, \cref{line:low-high} in \cref{alg:bcc}).
This step generates the tags used in the algorithm, including $w_1[\cdot]$, $w_2[\cdot]$, $\low[\cdot]$, $\high[\cdot]$, $\first[\cdot]$, $\last[\cdot]$
(same as in Tarjan-Vishkin, see \cref{sec:previous-bcc}) and the parent array $p[\cdot]$.
$\low[\cdot]$ and $\high[\cdot]$ values are computed by looping over all edges and getting arrays $w_1$ and $w_2$, and applying $n$ 1D range-minimum queries (RMQ).
This step takes in $O(n+m)$ work and $O(\log n)$ span~\cite{blelloch2020optimal}.
These tags will help to decide the four edge types (see details below).
All the tag arrays have size $O(n)$.


\myparagraph{\lastcc{}} (Step 4 in \cref{fig:bcc}, \cref{line:lastcc}--\ref{line:component-head} in \cref{alg:bcc}).
As mentioned, our \sketch{} graph $G'$ contains \plain{} tree edges and cross edges.
To achieve space efficiency, we do not explicitly store~$G'$.
Since $G'$ is a subgraph of $G$, we can directly use $G$ but skip the \fence{} edges and back edges, which can be determined using the tags generated in Step 3 (\cref{line:bcc-cond}--\ref{line:bcc-backward}). 
Then we compute the CCs on the \sketch{} $G'$ (\cref{line:lastcc}), which assigns a label $l[v]$ to each vertex (\cref{fig:bcc}, Step 4.1).
In \cref{lem:cctobcc1},
we show that if two vertices are connected in $G'$,
they must be biconnected on the input graph $G$.
We then assign the head to each label (\cref{line:component-head,line:component-head0}) by looping over all \fence{} edges (\cref{fig:bcc}, Step 4.2).
For a \fence{} edge $u$--$p(u)$, if $u$ and $p(u)$ have different labels (\cref{line:component-head0}),
$p(u)$ (intuitively) isolates vertices below $u$
with the other parts in the graph. Thus, we assign $p(u)$ as the component head of $u$'s CC in $G'$.
We prove the correctness of this step in \cref{lem:cctobcc2,lem:head-to}.
This step also only requires $O(n)$ \aux{} space, which is needed by running CC on $G$ but skip certain edges.

\begin{figure*}[t]
  \centering
  \vspace{-1.0em}
  \includegraphics[width=2.1\columnwidth]{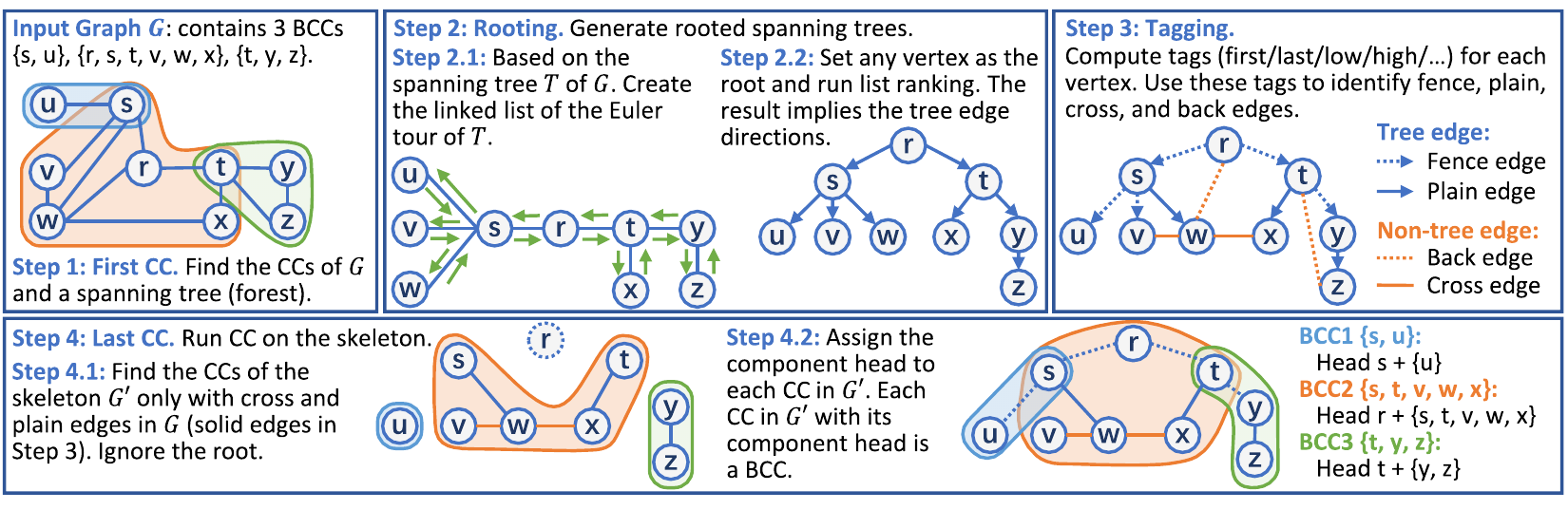}
  \vspace{-1.2em}
  \caption{\textbf{The outline of the \ouralgo algorithm and a running example.}  The four steps are explained in detail in \cref{sec:bcc-details}. \label{fig:bcc}}
  \vspace{-1.5em}
  \Description[<short description>]{<long description>}
\end{figure*}

\subsection{Correctness for the \Ouralgo{} Algorithm}\label{sec:correctness}
We now prove the correctness of our algorithm.
Note that our algorithm will identify the spanning forest in the first step and deal with each CC respectively.
For simplicity, throughout the section, we focus on one CC in $G$.

In the following, when we use the concepts about a spanning tree of the graph (e.g., {root}, {parent}, {child}, and {subtree}),
we refer to the specific spanning tree identified in Step 1 of our algorithm, and use $T$ to represent it.
Recall that ${T_u}$ denotes the subtree rooted at vertex $u$, and $u$\treepath$v$ denotes the tree path on $T$ from $u$ to $v$.
Some other notation is given in \cref{tab:notation}.
In a spanning tree, we say a node $u$ is shallower (deeper) than $v$ if $u$ is closer (farther) to the root than $v$.
We use node and vertex interchangeably.

We note that although \cref{alg:bcc} is simple, the correctness proof is sophisticated.
We show the relationship of facts, lemmas, and theorems in \cref{fig:bcc}. \ifconference{Due to the space limit, the proofs for \cref{lem:common,lem:cycle,lem:bcc-connected,lem:art-head,lem:function} are given in the full version of the paper, and we mainly focus on the proofs that reflect some key ideas in our new algorithm.}\iffullversion{
The proofs for \cref{lem:common,lem:cycle,lem:bcc-connected,lem:art-head,lem:function} are given in \cref{sec:add-proofs}, and here we mainly focus on the proofs that reflect some key ideas in our new algorithm.}

We first show some facts for BCCs based on the definition.

\begin{fact}\label{lem:common}
Two BCCs share at most one common vertex.
\end{fact}

\begin{fact}\label{lem:cycle}
For a cycle in a graph, all vertices on the cycle are in the same BCC.
\end{fact}

\begin{lemma}\label{lem:bcc-connected}
Given a graph $G$, vertices in each BCC $C\subseteq V$ must also be connected in an arbitrary spanning tree $T$ for $G$.
\end{lemma}

Since each BCC $C$ must be connected in the spanning tree in $T$, there must exist a unique shallowest node in this BCC on $T$.
We call this shallowest node the \defn{\bcchead{}} of the BCC~$C$, and denote it as $h_C$.  

\begin{lemma}\label{lem:art-head}
Each non-root \bcchead{} is an articulation point. An articulation point must be a \bcchead{}.
\end{lemma}

\begin{lemma}\label{lem:function}
The function \insketch{} (\cref{line:bcc-cond}) in \cref{alg:bcc} can correctly skip the fence and back edges.
\end{lemma}

\begin{figure}[t]
  \centering
  \includegraphics[width=\columnwidth]{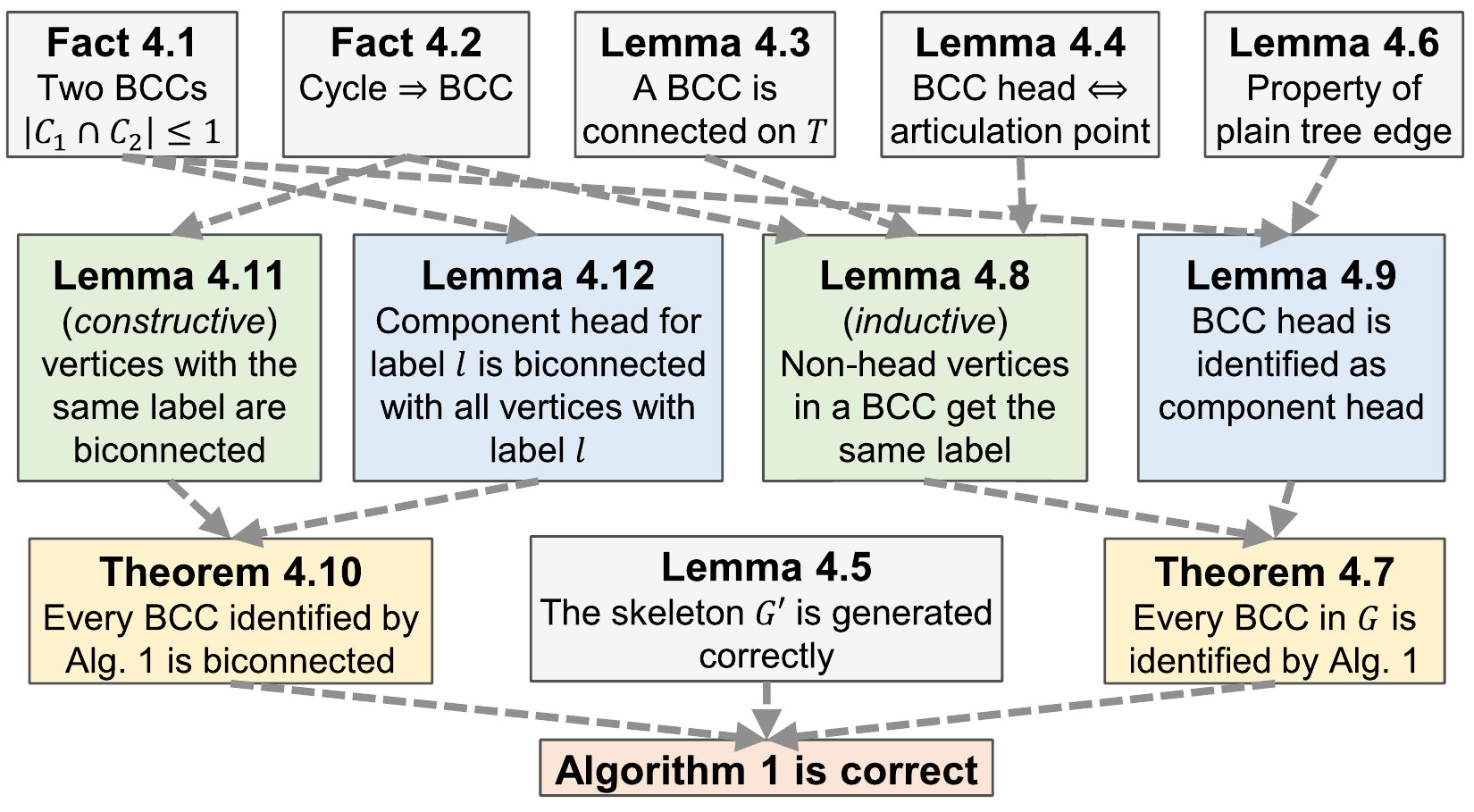}
  \caption{The structure of the correctness proof for \cref{alg:bcc}.  
  }\label{fig:proof}
  \Description[<short description>]{<long description>}
\end{figure}

Next, we show a useful property of the \plain{} tree edges.

\begin{lemma}\label{lem:non-critical}
For a \plain{} tree edge $x$--$y$ where $x$ is the parent of $y$, let $z$ be $x$'s parent, then $x,y,z$ are biconnected.
\end{lemma}
\begin{proof}
  Since $x$--$y$ is not a \fence{} edge, there must be an edge $a$--$b$, s.t.\ $a\in T_y$ and $b\notin T_x$.
  The cycle $y$\treepath$a$--$b$\treepath$z$--$x$--$y$ then contains $x$, $y$, and $z$.
  Due to \cref{lem:cycle}, $x$, $y$, and $z$ are in the same BCC.
\end{proof}

Next, we show that \cref{alg:bcc} can correctly identify all BCCs.
We will show two directions.
First, if two vertices $u$ and $v$ are biconnected, \cref{alg:bcc} must put them in a BCC.
Second, for any two vertices $u$ and $v$ in a BCC found by \cref{alg:bcc}, they must be biconnected.

\begin{theorem}\label{thm:algo-to-bcc}
For $u,v\in V$, if they are biconnected, \cref{alg:bcc} assigns them to the same BCC.
\end{theorem}
To prove \cref{thm:algo-to-bcc}, we discuss two cases: 1) one of $u$ and $v$ is a \bcchead{}, and 2) neither of them is a \bcchead{}.

\begin{lemma}\label{lem:nonhead-to}
For a BCC $C$ and two vertices $u,v\in C\setminus\{h_C\}$, they are connected in the \resgraph{} $G'$ and will get the same label in \cref{alg:bcc}.
\end{lemma}
\begin{proof}
  If all tree edges connecting $C\setminus\{h_C\}$ are \plain{} tree edges, $u$ and $v$ are already connected in $G'$.
  Next, we show that the two endpoints of every \fence{} edge are also connected in $G'$.
  To do so, we first sort (only conceptually) all vertices in $C\setminus\{h_C\}$ by their depth in $T$.
  Then we inductively show from bottom up (deep to shallow) that, given a vertex $v\in C$, $T_v\cap C$ ($v$'s subtree in $C$) is connected in $G'$.

  The base case is the deepest vertices in $C\setminus\{h_C\}$.
  In this case, their subtree contains only one vertex so they are connected.

  We now consider the inductive step---if for all vertices with depth $\ge d$, their subtrees in $C$ are connected in $G'$,
  then for all vertices with depth $d-1$, their subtrees in $C$ are also connected in $G'$.
  Consider a vertex $u\in C\setminus\{h_C\}$ with depth $d-1$.
  If $u$ has only one child $v$ in $C$, then $u$--$v$ is a \plain{} tree edge since otherwise $v$'s subtree cannot escape $u$'s subtree and $u$ is an articulation point (disconnecting $v$ and $p(u)$), contradicting \cref{lem:art-head}.
  Assume $u$ has multiple children $c_1,\ldots, c_k$ in $C$.
  Let $u$--$v$ be a \fence{} edge that is not in $G'$, where $v=c_i$ is a child of $u$.
  We will show that $u$ and $v$ are still connected in $G'$.

  Since $u$ is not a \bcchead{}, $p(u)$ must also be in $C$.
  Based on the definition of BCC, if we remove $u$,
  $v$ and $p(u)$ are still connected $C$.
  Let the path be $P=v$--$x_1$--$x_2$--...--$x_k$--$p(u)$ where $x_i\in C$ and $x_i\ne u$.
  We will construct a path in $G'$ from $P$ that connects $v$ and $u$.
  Let $x_{j+1}$ be the first vertex on path $P$ that is not in $T_u$. We will use the path $v=x_0$--$x_1$--$x_2$--...--$x_j$.
  All nodes in this path have depths $\ge d$.
  Due to the induction hypothesis, if some of the edges are back or \fence{} edges, we can replace them with the paths in $G'$, and denote this path as $P'$.
  Then, since $x_{j+1}\notin T_u$ is connected to $x_j\in T_u$, all edges on tree path $x_j$\treepath$u$ are \plain{} tree edges.
  As a result, $u$ and $v$ are connected in $G'$ using the path $P'$ from $v$ to $x_j$, and the tree path from $x_j$ to $u$ (all edges are in~$G'$).
  By the induction, all vertices in $C\setminus\{h_C\}$ are connected in $G'$, and hence get the same label after \cref{line:lastcc}.
\end{proof}

\begin{lemma}\label{lem:head-to}
Any \bcchead{} will be correctly identified as a component head in \cref{alg:bcc}.
\end{lemma}
\begin{proof}
  Consider a BCC $C$ and its \bcchead{} $h_C$. Among all the children of $h_C$, a subset $S$ of them are in the same BCC $C$. Consider any $c\in S$.
  We will show that the edge $c$--$h_C$ must be identified correctly in \cref{line:component-head0}. 

  We first show that $c$--$h_C$ must be a \fence{}.
  If $h_C$ is the root of $T$, and in this case, all tree edges connecting to $h_C$ are \fence{} edges.
  Otherwise, this can be inferred from the contrapositive of \cref{lem:non-critical}.
  If $c$--$h_C$ is a \plain{} tree edge, $c$, $h_C$, and $p(h_C)$ must be biconnected, which means $p(h_C)$ is also in the BCC $C$.
  This contradicts the assumption that $h_C$ is the shallowest node (\bcchead{}) in the BCC.

  We then show that after we run the CC on the \resgraph{} $G'$ (\cref{line:lastcc}), $h_C$ and $c$ have different labels (i.e., $h_C$ and $c$ are not connected in~$G'$).
  Assume to the contrary that there exists a path $P$ from $c$ to $h_C$ on $G'$.
  Consider the last node $t$ on the path before $h_C$.
  Because $h_C$--$c$ is a fence edge and is ignored in $G'$, $c\ne t$.
  We discuss three cases.
  (1) $t$ is not in the $h_C$'s subtree $T_{h_C}$.
  Consider the first edge $x$--$y$ on the path $P$ such that $x\in T_{h_C}$ and $y\ne T_{h_C}$.
  Since $x$--$y$ escapes $h_C$'s subtree, the tree path $P'=x\,$\textasciitilde$\,h_C$ only contains \plain{} tree edges.
  Let $c'$ be $h_C$'s child on the path $P'$.
  From \cref{lem:non-critical}, $c'$, $h_C$, and $p(h_C)$ are biconnected.
  In this case, $h_C$--$c\,$\textasciitilde$\,x\,$\textasciitilde$\,c'$--$\,h_C$ is a cycle, and \cref{lem:cycle} shows that $c'$, $h_C$ and $c$ are biconnected.
  The contrapositive of \cref{lem:common} indicates that $c'$, $h_C$, $c$, and $p(h_C)$ are all biconnected, contradicting the assumption that $h_C$ is the \bcchead{} (the shallowest node in the BCC).
  (2) $t\in T_{h_C}$, but $t$ is not $h_C$'s child. This is impossible because $t$--$h_C$ is a back edge, which is not in $G'$.
  (3) $t$ is a child of $h_C$.
  This case is similar to (1). By replacing $c'$ in the previous proof by $t$, we can get the same contradiction.
  Combining all cases proves that there is no path in $G'$ between $h_C$ and its children in $C$,
  so $l[h_C]$ is different from the labels of its children in $C$.
\end{proof}

Combining \cref{lem:head-to,lem:nonhead-to}, we can prove \cref{thm:algo-to-bcc}.

We then show the other direction---all the BCCs computed by \cref{alg:bcc} are indeed biconnected.

\begin{theorem}\label{thm:algo-from-bcc}
  If two vertices $u$ and $v$ are identified as in the same BCC by \cref{alg:bcc}, they must be biconnected.
\end{theorem}

Similar to the previous proof, we consider two cases:
(1) none of the two vertices is a component head (they are connected in $G'$), proved in \cref{lem:cctobcc1},
and
(2) one of them is identified as a component head in \cref{line:component-head}, proved in \cref{lem:cctobcc2}.

\begin{lemma}\label{lem:cctobcc1}
  If two vertices $u$ and $v$ are connected in the \resgraph{} $G'$, they are biconnected.
\end{lemma}
\begin{proof}
  Since $u$ and $v$ are connected in $G'$,
  there exists a path $P$ from $u$ to $v$ only using edges in $G'$.
  Let $P$ be $u=p_0$--$p_1$--...--$p_{k-1}$--$p_k=v$.
  We will show that after removing any vertex $p_i$ where $1\le i<k$ on $P$,
  $p_{i-1}$ and $p_{i+1}$ are still connected,
  meaning that $u$ and $v$ are biconnected.
  We summarize all possible local structures in three cases,
  based on whether $p_{i-1}$ (and $p_{i+1}$) is a child of $p_i$ in $T$.

  Case 1: both $p_{i-1}$ and $p_{i+1}$ are $p_i$'s children.
  Since $p_{i-1}$--$p_i$ is not a \fence{} edge, there must be an edge $x$--$y$ s.t. $x\in T_{p_{i-1}}$ and $y\notin T_{p_i}$.
  Similarly, for $p_i$--$p_{i+1}$, there exists an edge $(x',y')$ s.t. $x'\in T_{P_{i+1}}$ and $y'\notin T_{P_{i}}$.
  Hence, without using $p_i$, $p_{i-1}$ and $p_{i+1}$ are still connected by the path $p_{i-1}\,$\textasciitilde$\,x$--$y\,$\textasciitilde$\,y'$--$x'\,$\textasciitilde$\,p_{i+1}$.
  Here since $y,y'\notin T_{p_{i}}$, $y\,$\textasciitilde$\,y'$ does not contain $p_i$.


  Case 2: one of $p_{i-1}$ and $p_{i+1}$ is $p_i$'s child.
  WLOG, assume $p_{i-1}$ is the child.
  Since $p_{i-1}$--$p_i$ is not a \fence{} edge, there must be an edge $x$--$y$ such that $x\in T_{p_{i-1}}$ and $y\notin T_{p_i}$.
  Also, since $p_{i+1}$ is either the parent of $p_i$ or connected to $p_i$ using a cross edge, $p_{i+1}\notin T_{p_i}$.
  Hence, without using $p_i$, $p_{i-1}$ and $p_{i+1}$ are still connected using the path $p_{i-1}\,$\textasciitilde$\,x$--$y\,$\textasciitilde$\,p_{i+1}$.

  Case 3: neither $p_{i-1}$ nor $p_{i+1}$ is a child of $p_i$,
  and neither of them is in $T_{p_i}$ (otherwise they are connected by a back edge).
  Without using $p_i$, $p_{i-1}$ and $p_{i+1}$ are still connected using the tree path $p_{i-1}\,$\textasciitilde$\,p_{i+1}$.

  Since removing any vertex on the path $P$ does not disconnect the path, all vertices in the same CC of the \resgraph{} are biconnected.
\end{proof}

\begin{lemma}\label{lem:cctobcc2}
  If \cref{line:component-head} in \cref{alg:bcc} assigns $h$ as the component head of a connected component (CC) $C$ in the \resgraph{} $G'$, then $h$ is biconnected with $C$.
\end{lemma}
\begin{proof}
  First of all, assume $h$ is assigned as the component head because of its child $c$, where $h$--$c$ is a \fence{} edge.
  We will show that the connected component $C$ in $G'$ containing $c$ is biconnected with $h$.
  There are two cases.

  Case 1: $C$ only contains vertices in $T_c$.
  This means that no vertices in $T_c$ have a cross edge to another vertex outside $T_c$.
  Therefore, either all edges incident on $c'\in T_c$ do not escape from $T_c$,
  or some node $c'\in T_c$ is connected to nodes outside $T_c$ via back edges.
  In the former case, all the edges connecting $c$ and its children are \fence{} edges,
  and thus $C$ only contains $c$. 
  In this case, $h$ is trivially biconnected with $C$.
  In the latter case, assume $x\in T_c\cap C$ has a back edge connected to $y\notin T_c$.
  Note that $y$ can only be $h$---if $y$ is $h$'s ancestor, then edge $x$--$y$ escapes $T_h$, so $h$--$c$ is a \plain{} tree edge (contradiction).
  Therefore, we can find a cycle $h$--$c\,$\textasciitilde$\,x$--$h$.
  From \cref{lem:cycle}, $h,c,x$ are biconnected, and $h$ is in the same BCC as $c$ and $x$,
  and thus all vertices in $C$ (\cref{lem:cctobcc1,lem:common}).

  Case 2: $C$ contains both vertices in $T_c$ and some vertices in $T_h\setminus T_c$.
  Hence, there exists a cross edge $x$--$y$, where $x\in T_c$ and $y\notin T_c$.
  We can find a cycle $h,$\textasciitilde$\,x$--$y\,$\textasciitilde$\,h$.
  From \cref{lem:cycle}, $h,c,u$ are biconnected. $h$ is in the same BCC as $c$ and $u$.
\end{proof}

Combining \cref{lem:cctobcc1,lem:cctobcc2} proves \cref{thm:algo-from-bcc}.

\cref{thm:algo-to-bcc} shows that if two vertices are put in the same BCC by \cref{alg:bcc}, they are biconnected in $G$.
\cref{thm:algo-from-bcc} indicates that two vertices biconnected in $G$ will be put in the same BCC by \cref{alg:bcc}.
\cref{lem:function} indications back edges and \fence{} edges are identified correctly by \cref{alg:bcc}.
Combining them together indicates that \cref{alg:bcc} is correct.

\hide{
\begin{figure}[t]
    \includegraphics[width=\columnwidth]{figures/BCC-proof-new.pdf}
    \caption{\small An example for different cases discussed in the proof.  Dotted arrows are critical edges.  Blue straight lines are non-critical tree edges.  Orange curve lines are cross edges.  The entire graph is biconnected, and we want to show that all non-root vertices are connected via non-critical tree edges and cross edges. \label{fig:BCC-proof}}
\end{figure}
}

\subsection{Cost Bounds for the \Ouralgo{} Algorithm}\label{sec:cost}

We now analyze the cost bounds of the algorithm.
\begin{theorem}\label{thm:bcccost}
    \cref{alg:bcc} computes the BCCs of a graph $G$ with $n$ vertices and $m$ edges using $O(n+m)$ expected work, $O(\log^3 n)$ span \whp{}, and $O(n)$ auxiliary space (other than the input).
\end{theorem}

\begin{proof}
    The first and last steps compute the graph connectivity twice.
    Graph connectivity can be computed in $O(n+m)$ expected work and $O(\log^3 n)$ span \whp{}~\cite{SDB14}.
    In Step 2, ETT can be performed $O(n)$ expected work and $O(\log n)$ span \whp{} (see \cref{sec:prelim}).
    In Step 3, computing $\low[\cdot]$ and $\high[\cdot]$ arrays based on RMQ takes $O(m)$ work and $O(\log n)$ span~\cite{blelloch2020optimal}.
    Adding all pieces together gives the work and span bounds.

    For the space, all arrays for the tags have size $O(n)$.
    As mentioned, we do not generate the \resgraph{} explicitly.
    In the last step, we try all the edges in $G$ but skipping the back and \fence{} edges.
    In all, the auxiliary space needed is $O(n)$.
  \end{proof}

\section{Implementation Details}\label{sec:implementation}
We discuss some implementation details of \ouralgo{} in this section.

\myparagraph{Connectivity.}
\label{sec:imp:cc}
Connectivity is used twice in \ouralgo{}.
The only existing parallel CC implementation with good theoretical guarantee we know of is the SDB algorithm~\cite{SDB14}
(an initial version of \gbbs{} is based on this algorithm).
A recent paper by Dhulipala et al.~\cite{dhulipala2020connectit} gave 232 parallel CC implementations, many of which outperformed the SDB algorithm,
but no analysis of work-efficiency was given.
A more recent version of \gbbs{} uses the \texttt{UF-Async} algorithm in~\cite{dhulipala2020connectit} to compute CC.
To achieve efficiency both in theory and in practice, \ouralgo uses the \defn{LDD-UF-JTB} algorithm from~\cite{dhulipala2020connectit} and we provide a new analysis for this algorithm to prove its theoretical efficiency.

LDD-UF-JTB consists of two steps.
It first runs a {low-diameter decomposition} (LDD) algorithm~\cite{miller2013parallel} to find a decomposition (partition of vertices) of the graph
such that each component has a low diameter and the number of edges crossing different components is bounded.
The second step is to use a {union-find structure} by Jayanti et al.~\cite{jayanti2019randomized} to union components connected by cross-component edges.
We now show the bounds of this algorithm.

\begin{theorem}\label{thm:connect}
    The LDD-UF-JTB algorithm computes the CCs of a graph $G$ with $n$ vertices and $m$ edges using $O(n+m)$ expected work and $O(\log^3 n)$ span \whp.
\end{theorem}

Therefore, using LDD-UF-JTB for CC preserves the cost bounds in \cref{thm:bcccost}. \ifconference{We prove \cref{thm:connect} in the full version of this paper.}\iffullversion{We prove \cref{thm:connect} in \cref{app:cc}.}

We optimized LDD-UF-JTB using the hash bag and local search techniques proposed from~\cite{scc2022}.
These optimizations are only used in computing CCs in our algorithm,
and we do not claim them as contributions of this paper.
In our tests, using these optimizations improves the performance of \ouralgo{} by 1.5$\times$ on average (up to 5$\times$).
Some results are shown \ifconference{in the full version of this paper.}\iffullversion{in \cref{fig:BCC_Break_loc}.}We note that among all 232 CC algorithms in~\cite{dhulipala2020connectit}, no one is constantly faster, and the relative performance is decided by the input graph properties.
In \ouralgo, we currently use the same CC algorithm for all graphs, and we acknowledge that using the fastest CC algorithm on each graph can further improve the performance of \ouralgo. We choose LDD-UF-JTB mainly because it is theoretically-efficient, and also can generate CC as a by-product efficiently.

\hide{The LDD algorithm runs in rounds.
If the frontier (vertices being processed in one round) size is small, to gain parallelism,
the local search technique allows the algorithm to explore multi-hop neighbors instead of one-hop neighbors of the frontier.
This can be viewed as ``vertical'' granularity control, which is similar to setting a sequential base case size in some divide-and-conquer algorithms. \xiaojun{Give citations?}
The hash bag uses dynamic resizing to avoid the redundant work in maintaining frontier.
More details are given in the supplementary materials.
}

\myparagraph{Spanning Forest.}
The spanning forest of $G$ is obtained as a by-product of Step 1, which saves all edges to form the CCs.
We then re-order the vertices in the compressed sparse row (CSR) format to let each CC be contiguous.

\myparagraph{Euler Tour Technique (ETT).}
We use the standard ETT to root the spanning trees (see \cref{sec:prelim}).
We replicate each undirected edge in $T$ into two directed edges and semisort them~\cite{gu2015top}, so edges with the same first endpoint are contiguous.
Then we construct a circular linked list as the Euler circuit.
Assume a vertex $v$ has $k$ in-coming neighbors $u_1$, $u_2$, $\cdots$, $u_k$.
For every incoming edge of $v$ except for the last one,
we link it to its next outgoing edge
(i.e., $u_i$--$v$ is linked to $v$--$u_{i+1}$ for $1 \leq i < k$).
For the last incoming edge, we link it to the first outgoing edge of $v$
(i.e., $u_k$--$v$ is linked to $v$--$u_1$).

After we obtain the Euler circuit of the tree, we flatten the linked list to an array by list ranking, and acquire the Euler tour order of each vertex.
For list ranking, we coarsen the base cases by sampling $\sqrt{n}$ nodes.
We start from these nodes in parallel, with each node sequentially following the pointers until it visits the next sample.
Then we compute the offsets of each sample by prefix sum, pass the offsets to other nodes by chasing the pointers from the samples, and scatter all nodes into a contiguous array.

\myparagraph{Computing Tags.}\label{app:bcc:compute}
We use several tags $w_1$, $w_2$, $\first$, $\last$, $\low$, and $\high$ for each vertex, defined the same as Tarjan-Vishkin~\cite{tarjan1985efficient}  (see~\cref{sec:previous-bcc}).
{
We use \CAS{} operations to compute $\first$ and $\last$ as they represent the first and last appearances of a vertex in the Euler tour order.
For each tree edge $(u, v)$, if $\first[u] < \first[v]$, we set $p(v)=u$, or vice versa.
}
Computing $\low$ and $\high$ are similar, so we only discuss $\low$ here.
We first initialize $w_1[v]$ with $\first[v]$ for each $v \in V$.
Then it traverses all non-tree edges $u$--$v$ and updates $w_1[u]$ and $w_1[v]$ with the minimum of $\first[u]$ and $\first[v]$.
We build a parallel sparse table~\cite{blelloch2020optimal} on $w_1$ to support range minimum queries.
Note that $\first[v]$ and $\last[v]$ reflect the range of $v$'s subtree in the Euler tour order.
Thus, $\low[v]$ can be computed by finding the minimum element in $w_1[\cdot]$ in the range between $\first[v]$ and $\last[v]$.
$\high[\cdot]$ can be computed similarly.

\begin{table*}[htbp]
    \small
    \centering
    \setlength\tabcolsep{2.5pt}
    \begin{tabular}{@{}c@{ }r|rr@{ }r@{ }r@{}r|ccc|ccc|cc|c|l@{}}
        &       & \multicolumn{1}{c}{\multirow{2}[1]{*}{\textbf{$\boldsymbol{n}$}}} & \multicolumn{1}{c}{\multirow{2}[1]{*}{\textbf{$\boldsymbol{m}$}}} & \multicolumn{1}{c}{\multirow{2}[1]{*}{\textbf{$\boldsymbol{D}$}}} & \multicolumn{1}{c}{\multirow{2}[1]{*}{\textbf{\textbf{\#BCC}}}} & \multicolumn{1}{@{}c@{ }|}{\multirow{2}[1]{*}{\textbf{$\boldsymbol{|\textbf{BCC}_1|\%}$}}} & \multicolumn{3}{c|}{\textbf{Ours}} & \multicolumn{3}{c|}{\textbf{GBBS}} & \textbf{SM'} & \multirow{2}[1]{*}{\textbf{SEQ}} & \cellcolor[rgb]{ 1,  1,  0} \textbf{$\boldsymbol{T_{\mathit{best}}}$} & \multicolumn{1}{c}{\multirow{2}[1]{*}{\textbf{Notes}}} \\
        &       &       &       &       &       &       & \textbf{par.} & \textbf{seq.} & \textbf{spd.} & \textbf{par.} & \textbf{seq.} & \textbf{spd.} & \textbf{14} &       & \cellcolor[rgb]{ 1,  1,  0} \textbf{/ours} &  \\
    \midrule
    \multirow{5}[2]{*}{\begin{sideways}\textbf{Social}\end{sideways}} & \textbf{YT} & 1.13M & 5.98M & 23    & 673,661 & 39.83\% & \underline{0.030} & 0.465 & 15.6  & 0.040 & 0.435 & 10.8  & 0.059 & 0.175 & \cellcolor[rgb]{ 1,  1,  0} 1.35 & com-youtube~\cite{yang2015defining} \\
        & \textbf{OK} & 3.07M & 234M  & 9     & 68,117 & 97.76\% & \underline{0.103} & 3.08  & 30.0  & 0.158 & 4.86  & 30.8  & 0.297 & 3.14  & \cellcolor[rgb]{ 1,  1,  0} 1.53 & com-orkut~\cite{yang2015defining} \\
        & \textbf{LJ} & 4.85M & 85.7M & 19    & 1,133,883 & 75.61\% & \underline{0.104} & 3.02  & 28.9  & 0.159 & 3.34  & 21.0  & n     & 1.87  & \cellcolor[rgb]{ 1,  1,  0} 1.52 & soc-LiveJournal1~\cite{backstrom2006group} \\
        & \textbf{TW} & 41.7M & 2.41B & 23    & 1,936,001 & 95.33\% & \underline{1.44} & 52.9  & 36.7  & 2.83  & 95.2  & 33.7  & 20.5$^*$  & 49.2  & \cellcolor[rgb]{ 1,  1,  0} 1.96 & Twitter~\cite{kwak2010twitter} \\
        & \textbf{FT} & 65.6M & 3.61B & 37    & 14,039,045 & 78.50\% & \underline{3.10} & 129   & 41.6  & 6.44  & 260   & 40.5  & 10.9  & 122   & \cellcolor[rgb]{ 1,  1,  0} 2.07 & Friendster~\cite{yang2015defining} \\
    \midrule
    \multirow{5}[2]{*}{\begin{sideways}\textbf{Web}\end{sideways}} & \textbf{GG} & 876K  & 8.64M & 24    & 175,274 & 73.31\% & \underline{0.029} & 0.534 & 18.7  & 0.045 & 0.530 & 11.8  & n     & 0.255 & \cellcolor[rgb]{ 1,  1,  0} 1.58 & web-Google~\cite{leskovec2009community} \\
        & \textbf{SD} & 89.2M & 3.88B & 35    & 16,189,065 & 80.36\% & \underline{3.11} & 134   & 43.2  & 5.61  & 213   & 38.0  & n     & 92.3  & \cellcolor[rgb]{ 1,  1,  0} 1.81 & sd\_arc~\cite{webgraph} \\
        & \textbf{CW} & 978M  & 74.7B & 254   & 81,809,602 & 86.48\% & \underline{22.9} & 1464  & 64.0  & 39.7  & 1526  & 38.4  & n     & 695   & \cellcolor[rgb]{ 1,  1,  0} 1.73 & ClueWeb~\cite{webgraph} \\
        & \textbf{HL14} & 1.72B & 124B  & 366   & 124,406,075 & 83.25\% & \underline{31.1} & 2057  & 66.0  & 50.7  & 2113  & 41.7  & n     & 1011  & \cellcolor[rgb]{ 1,  1,  0} 1.63 & Hyperlink14~\cite{webgraph} \\
        & \textbf{HL12} & 3.56B & 226B  & 650   & 410,853,262 & 80.63\% & \underline{89.1} & 5435  & 61.0  & 104   & 5985  & 57.6  & n     & 3027  & \cellcolor[rgb]{ 1,  1,  0} 1.17 & Hyperlink12~\cite{webgraph} \\
    \midrule
    \multirow{3}[2]{*}{\begin{sideways}\textbf{Road}\end{sideways}} & \textbf{CA} & 1.97M & 5.53M & 857   & 381,366 & 79.55\% & \underline{0.040} & 0.824 & 20.6  & \textcolor[rgb]{ 1,  0,  0}{0.372} & 1.05  & 2.82  & n     & 0.206 & \cellcolor[rgb]{ 1,  1,  0} 5.15 & roadnet-CA~\cite{leskovec2009community} \\
        & \textbf{USA} & 23.9M & 57.7M & 8,263 & 7,390,330 & 66.90\% & \underline{0.336} & 12.1  & 36.0  & \textcolor[rgb]{ 1,  0,  0}{4.64} & 15.1  & 3.25  & \textcolor[rgb]{ 1,  0,  0}{3.73$^*$} & 2.25  & \cellcolor[rgb]{ 1,  1,  0} 6.69 & RoadUSA~\cite{roadgraph} \\
        & \textbf{GE} & 12.3M & 32.3M & 2,240 & 2,482,488 & 78.67\% & \underline{0.267} & 11.1  & 41.6  & 2.02  & 11.4  & 5.66  & 1.14$^*$ & 2.88  & \cellcolor[rgb]{ 1,  1,  0} 7.54 & Germany~\cite{roadgraph} \\
    \midrule
    \multirow{8}[2]{*}{\begin{sideways}\boldmath{}\textbf{$\boldsymbol{k}$-NN}\unboldmath{}\end{sideways}} & \textbf{HH5} & 2.05M & 13.0M & 1,859 & 17,408 & 62.55\% & \underline{0.073} & 1.60  & 22.0  & 0.447 & 1.52  & 3.41  & n     & 0.509 & \cellcolor[rgb]{ 1,  1,  0} 6.16 & Household~\cite{uciml,wang2021geograph}, $k$=5 \\
        & \textbf{CH5} & 4.21M & 29.7M & 14,479 & 299   & 15.41\% & \underline{0.128} & 2.85  & 22.2  & \textcolor[rgb]{ 1,  0,  0}{1.44} & 2.38  & 1.66  & n     & 0.528 & \cellcolor[rgb]{ 1,  1,  0} 4.11 & CHEM~\cite{fonollosa2015reservoir,wang2021geograph}, $k$=5 \\
        & \textbf{GL2} & 24.9M & 65.4M & 13,333 & 10,940,922 & 0.03\% & \underline{0.402} & 13.8  & 34.5  & 1.53  & 16.9  & 11.0  & n     & 2.51  & \cellcolor[rgb]{ 1,  1,  0} 3.80 & GeoLife~\cite{geolife,wang2021geograph}, $k$=2 \\
        & \textbf{GL5} & 24.9M & 157M  & 21,600 & 1,009,434 & 30.07\% & \underline{0.472} & 19.1  & 40.5  & 2.80  & 19.4  & 6.92  & n     & 4.03  & \cellcolor[rgb]{ 1,  1,  0} 5.93 & GeoLife~\cite{geolife,wang2021geograph}, $k$=5 \\
        & \textbf{GL10} & 24.9M & 305M  & 3,824 & 51,465 & 86.38\% & \underline{0.668} & 29.2  & 43.8  & 1.64  & 23.5  & 14.3  & n     & 7.07  & \cellcolor[rgb]{ 1,  1,  0} 2.46 & GeoLife~\cite{geolife,wang2021geograph}, $k$=10 \\
        & \textbf{GL15} & 24.9M & 453M  & 3,664 & 23,149 & 91.11\% & \underline{0.751} & 34.4  & 45.8  & 1.51  & 25.9  & 17.1  & n     & 8.92  & \cellcolor[rgb]{ 1,  1,  0} 2.01 & GeoLife~\cite{geolife,wang2021geograph}, $k$=15 \\
        & \textbf{GL20} & 24.9M & 602M  & 2,805 & 13,619 & 93.96\% & \underline{0.861} & 39.2  & 45.6  & 1.48  & 28.6  & 19.3  & n     & 10.2  & \cellcolor[rgb]{ 1,  1,  0} 1.72 & GeoLife~\cite{geolife,wang2021geograph}, $k$=20 \\
        & \textbf{COS5} & 321M  & 1.96B & 1,180 & 85,283 & 99.74\% & \underline{8.46} & 382   & 45.2  & 17.5  & 392   & 22.4  & n     & 120   & \cellcolor[rgb]{ 1,  1,  0} 2.07 & Cosmo50~\cite{cosmo50,wang2021geograph}, $k$=5 \\
    \midrule
    \multirow{6}[2]{*}{\begin{sideways}\textbf{Synthetic}\end{sideways}} & \textbf{SQR} & 100M  & 400M  & 10,000 & 1     & 100.00\% & \underline{1.32} & 43.4  & 32.9  & 15.4  & 44.2  & 2.87  & 20.3$^*$ & 24.4  & \cellcolor[rgb]{ 1,  1,  0} 11.7 & 2D grid $10^4 \times 10^4$ \\
        & \textbf{REC} & 100M  & 240M  & 50,500 & 1     & 100.00\% & \underline{1.35} & 43.6  & 32.4  & \textcolor[rgb]{ 1,  0,  0}{47.0} & 34.6  & 0.735 & 13.1$^*$ & 16.8  & \cellcolor[rgb]{ 1,  1,  0} 12.5 & 2D grid $10^3 \times 10^5$ \\
        & \textbf{SQR'} & 100M  & 400M  & 10,256 & 23,836,580 & 70.65\% & \underline{1.31} & 50.1  & 38.1  & \textcolor[rgb]{ 1,  0,  0}{12.5} & 60.9  & 4.88  & n & 10.6  & \cellcolor[rgb]{ 1,  1,  0} 8.06 & sampled SQR \\
        & \textbf{REC'} & 100M  & 240M  & 69,014 & 23,826,514 & 70.66\% & \underline{1.37} & 46.8  & 34.3  & \textcolor[rgb]{ 1,  0,  0}{22.4} & 58.9  & 2.63  & n & 10.7  & \cellcolor[rgb]{ 1,  1,  0} 7.81 & sampled REC \\
        & \textbf{Chn7} & 10M   & 20M   & $10^7-1$ & $10^7-1$ & 0.00\% & \underline{0.278} & 13.1  & 46.9  & \textcolor[rgb]{ 1,  0,  0}{81.6} & 19.7  & 0.241 & \textcolor[rgb]{ 1,  0,  0}{40.5$^*$} & 3.33  & \cellcolor[rgb]{ 1,  1,  0} 12.0 & Chain of size $10^7$ \\
        & \textbf{Chn8} & 100M  & 200M  & $10^8-1$ & $10^8-1$ & 0.00\% & \underline{3.25} & 152   & 46.9  & \textcolor[rgb]{ 1,  0,  0}{957} & 307   & 0.320 & \textcolor[rgb]{ 1,  0,  0}{703$^*$} & 38.9  & \cellcolor[rgb]{ 1,  1,  0} 12.0 & Chain of size $10^8$ \\
    \bottomrule
    \end{tabular}%
    \caption{\small\textbf{Graph information, running times (in seconds), and speedups.}\label{tab:bcc}
    $T_{\mathit{best}}/\mathit{ours}$ (highlighted in yellow) is the \textbf{fastest time of the other implementations $/$ our time, both using all cores}.
    ``$n$'' $=$ number of vertices.
    ``$m$'' $=$ number of edges.
    ``$D$'' $=$ approximate diameter.
    ``\#BCC'' $=$ number of BCCs.
    ``$|\text{BCC}_1|\%$'' $=$ percentage of the largest BCCs.
    ``\gbbs{}'' $=$ \gbbs{}'s implementation~\cite{gbbs2021}.
    ``\hipc{}'' $=$ Slota and Madduri's algorithm~\cite{slota2014simple} (the faster of the two proposed algorithms). Since \hipc has scalability issues (see \cref{fig:scalability}), we report the 16-core time if it is faster, and denote as ($^*$).
    ``\seq'' $=$ Hopcroft-Tarjan BCC algorithm~\cite{hopcroft1973algorithm}.
    Details about the baselines are introduced in \cref{sec:exp}.
    The fastest runtime for each graph is underlined.
    Red numbers are parallel runtime \emph{slower than the sequential algorithm}.
    ``par.'' $=$ parallel running time (on 192 hyper-threads).
    ``seq.'' $=$ sequential running time (on 1 thread).
    ``spd.'' $=$ self-relative speedup.
    ``n'' $=$ no support, because \hipc{} only works on connected graphs.
    }
    \vspace{-1em}
\end{table*}%
\newcommand{\largediam}{large-diameter}
\newcommand{\lowdiam}{low-diameter}
\section{Experiments}
\label{sec:exp}

\myparagraph{Setup.}
We run our experiments on a 96-core (192 hyperthreads) machine with four Intel Xeon Gold 6252 CPUs, and 1.5 TB of main memory.
We implemented all algorithms in C++ using ParlayLib~\cite{blelloch2020parlaylib} for fork-join parallelism and some parallel primitives (e.g., sorting).
We use \texttt{numactl -i all} in experiments with
more than one thread to spread the memory pages across CPUs in a round-robin fashion.
We run each test for 10 times and report the median.

We tested on 27 graphs, including social networks, web graphs, road graphs, \knn{} graphs, and synthetic graphs.
The information of the graphs is given in \cref{tab:bcc}.
In addition to commonly-used benchmarks of social, web and, road graphs, we also
use \knn{} graphs and synthetic graphs.
\knn{} graphs are widely used in machine learning algorithms (see discussions in~\cite{wang2021geograph}).
In \knn{} graphs, each vertex is a multi-dimensional data point and has $k$ edges pointing to its $k$-nearest neighbors (excluding itself).
We also create six synthetic graphs, including two grids (SQR and REC),
two sampled grids (SQR' and REC', each edge is created with probability 0.6),
and two chains (Chn7 and Chn8).
SQR and SQR' have sizes $10^4\times 10^4$.
REC and REC' have sizes $10^3 \times 10^5$.
Each row and column in grid graphs are circular.
Chn7 and Chn8 have sizes $10^7$ and $10^8$. 
The tested graphs cover a wide range of sizes and edge distributions.

For directed graphs, we symmetrize them to test \BCC{}.
We call the social and web graphs \emph{\lowdiam{} graphs} as they have diameters mostly within a few hundreds. We call the road, \knn{}, and synthetic graphs \emph{\largediam{} graphs} as their diameters are mostly more than a thousand.
When comparing the \emph{average} running times across multiple graphs, we always take the \defn{geometric mean} of the numbers.

\myparagraph{Baseline Algorithms.}
We call all existing algorithms that we compare to the \defn{baselines}.
We implement sequential Hopcroft-Tarjan~\cite{hopcroft1973algorithm} algorithm for comparison,
referred to as \seq{}.
We compare the number of BCCs reported by each algorithm with \seq{} to verify correctness.

We also compare to two most recent available BCC implementations \gbbs{}~\cite{gbbs2021}, and Slota and Madduri~\cite{slota2014simple}.
We use \hipc{} to denote the \emph{better} of the two BCC algorithms in Slota and Madduri~\cite{slota2014simple}.
On many graphs, we observe that \hipc{} is faster on 16 threads than using all 192 threads,
in which case we report the lower time of 16 and 192 threads.
Through correspondence with the authors,
we understand that \hipc{} requires the input graph to be connected,
so we only report the running time when it gives the correct answers.
As few graphs we tested are entirely connected, we focus on comparisons with \gbbs{} and \seq{}.
We also compare our breakdown and sequential running times with \gbbs{} since \gbbs{} can process most of the tested graphs\footnote{\gbbs{} updated a new version after this paper was accepted, so we also updated the numbers using their latest version (Nov. 2022). Some new features in the latest version greatly improved their BCC performance. }.

Unfortunately, we cannot find existing implementations for Tarjan-Vishkin to compare with.
We are aware of two papers that implemented Tarjan-Vishkin~\cite{edwards2012better,cong2005experimental}.
Edwards and Vishkin's implementation~\cite{edwards2012better} is on the XMT architecture and they did not release their code.
Cong and Bader's code~\cite{cong2005experimental} is released, but it was written in 2005 and uses some system functions that are no longer supported on our machine.
For a full comparison, we implemented a faithful Tarjan-Vishkin from the original paper~\cite{tarjan1985efficient}.
\ifconference{As engineering Tarjan-Vishkin is not the main focus of this paper, we provide the details in the full paper.
}
\iffullversion{As engineering Tarjan-Vishkin is not the main focus of this paper, we mainly use it to evaluate the memory usage.
}

We note that both \gbbs{} and \hipc{} exclude the postprocessing to compute the actual BCCs,
but only report the number of BCCs at the end of the algorithm.
We include this step in \ouralgo{}, although this postprocessing only takes at most 2\% of the total running time in all our tests.

\begin{figure}[t]
  \centering
  \includegraphics[width=\columnwidth]{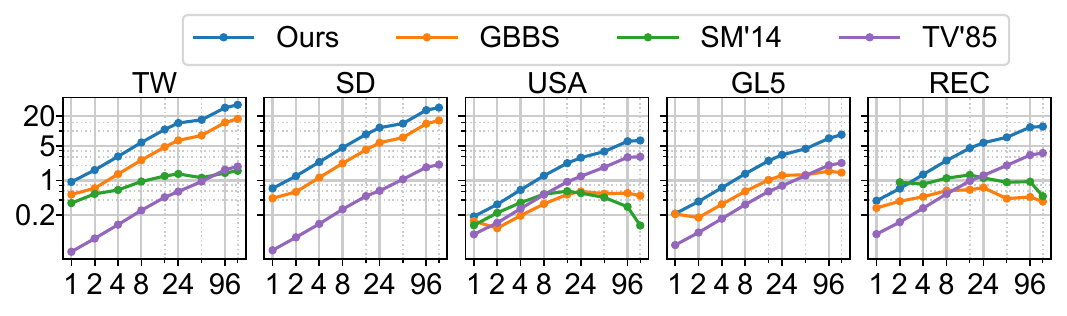}
  \caption{\small \textbf{Scalability curves for different BCC algorithms.} In each plot, $x$-axis is core counts (last data point is 96 core with hyperthreading) and $y$-axis is speedups normalized to SEQ (the sequential Hopcroft-Tarjan algorithm). Higher is better. SEQ is 1.}\label{fig:scalability}
  \Description[<short description>]{<long description>}
\end{figure} 

\subsection{Overall Performance}\label{sec:exp:bcc}
We present the running time of all algorithms in \cref{tab:bcc}.
Our \Ouralgo is \defn{faster than all baselines on all graphs}, mainly due to the theoretical efficiency---work- and space-efficiency enables competitive sequential times over the Hopcroft-Tarjan sequential algorithm, and polylogarithmic span ensures good speedup for all graphs.

\myparagraph{Sequential Running Time.}
We first compare the \emph{sequential} running time of
\seq{}, \gbbs{}, and \ouralgo. 
\seq{} and \ouralgo use $O(n+m)$ work. 
To enable parallelism,
both \ouralgo and \gbbs{} traverse all edges multiple times (running CC twice in Steps 1 and 4, and computing low/high for the \sketch{} in Step 3).
We describe more details about \gbbs{} implementation in \cref{sec:exp:breakdown}.
On average, our sequential time is 2.8$\times$ slower than \seq{}, but is 10\% faster than \gbbs{}.

\begin{figure*}[t]
  \centering
  \vspace{-.5em}
  \includegraphics[width=0.8\textwidth]{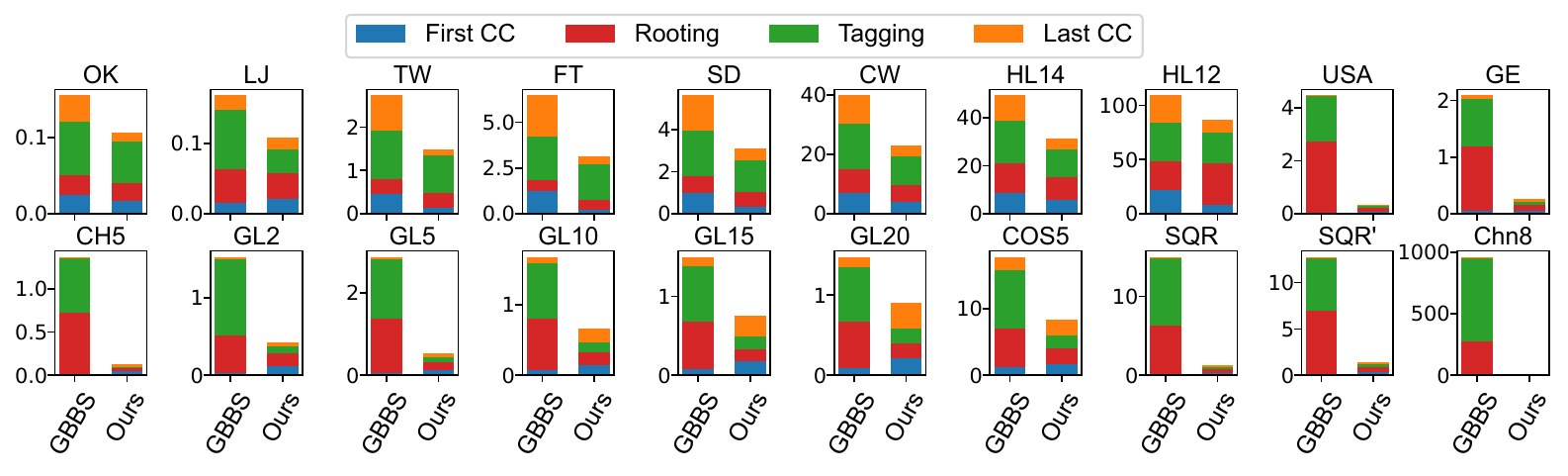}
  \vspace{-.5em}
  \caption{\small \textbf{BCC breakdown.} $y$-axis is the running time in seconds. The results for all the 27 graphs are in the full paper.}\label{fig:BCC_Total_Break}
  \vspace{-1.5em}
  \Description[<short description>]{<long description>}
\end{figure*}

\myparagraph{Scalability and Parallelism.} We report the scalability curves for \ouralgo{}, \gbbs{} and \hipc{} on five graphs (\cref{fig:scalability}).
For fair comparison, the speedup numbers in \cref{fig:scalability} are normalized to the running time of \seq{}.
On these graphs, \ouralgo is the only algorithm that scales to all processors.
It outperforms \gbbs{} and \hipc{} on all graphs with all numbers of threads (except REC on 2 cores).
We noticed that \hipc suffers from scalability issues, and the best performance can be achieved at around 16 threads.
Hence, we report \hipc's better running time of 16 and 192 threads in \cref{tab:bcc}.
\gbbs{} has similar issues on a few graphs.
However, as \gbbs{}'s performance does not drop significantly as core count increases, we consistently report \gbbs's time on 192 threads in \cref{tab:bcc}.

Our average self-relative speedups on both \lowdiam{} graphs and \largediam{} graphs are 36$\times$.
On large-scale low-diameter graphs with sufficient parallelism, the self-relative speedup can be up to 66$\times$.
Even on large-diameter graphs, \ouralgo{} achieves up to 47$\times$ self-relative speedup.
In comparison, the self-relative speedup of \gbbs's BFS-based algorithm is 29$\times$ on \lowdiam{} graphs and 3.7$\times$ on \largediam{} graphs.
This makes \gbbs{} only 11\% faster than \seq{} on \largediam{} graphs (and can be slower on some graphs), while ours is 5.1--18.5$\times$ better.
Overall, our parallel running time is 10$\times$ faster on \largediam{} graphs and 1.6$\times$ faster on \lowdiam{} graphs than \gbbs{}.
On some graphs, \hipc{} achieves better performance than \gbbs{}, but \ouralgo{} is 1.7--11.1$\times$ faster than \hipc{} on all the graphs.

To verify that \gbbs's performance is bottlenecked by BFS, we created \knn{} graphs GL2--20 from the set of points but with different values of $k$.
When increasing $k$ over 5, the graphs have more edges but smaller diameters.
For both \ouralgo and \seq, the running times increase when $k$ grows due to more edges (and thus more work), but the trend of \gbbs's running time is decreasing.
This indicates that the BFS is the dominating part of running time for \gbbs{}, and the performance on \gbbs{} is bottlenecked by the $O(\textsc{Diam}(G)\log n)$ span.


\subsection{Performance Breakdown}\label{sec:exp:breakdown}
To understand the performance gain of \ouralgo over prior parallel BFS-based BCC algorithms,
we compare our performance breakdown with \gbbs{} in \cref{fig:BCC_Total_Break}.
We choose \gbbs{} because it can process all graphs.
Since \gbbs{} is also in the \sketchconnect{} framework, we use the same four step names for \gbbs{} as in \ouralgo,
but there are a few differences.
(1) For \emph{First-CC}, \Ouralgo generates a spanning forest while \gbbs{} only finds all CCs.
(2) For \emph{\rootst{}}, \Ouralgo uses ETT to root the tree while \gbbs{} applies BFS on all CCs to find the spanning trees.
(3) The task for \emph{\lowhigh} is almost the same, but \gbbs{} computes fewer tags than \ouralgo{} since it is based on BFS trees.
\Ouralgo uses 1D RMQ queries that are theoretically-efficient, while \gbbs uses a bottom-up traversal on the BFS tree.
(4) For \emph{Last-CC}, both algorithms run CC algorithms on the \sketches{} to find BCCs.

We first discuss the two steps \firstcc{} and \lastcc{} that use connectivity.
\gbbs{} can be faster than \ouralgo{} in \firstcc{} on some graphs.
The reason is that our algorithm also constructs the spanning forest in \firstcc{}, while \gbbs{} has to run BFS in \rootst{} to generate the BFS spanning forest.
In \lastcc{}, the two algorithms achieve similar performance, and in many cases, \ouralgo{} is faster.
We note that the CC algorithm is independent with the BCC algorithm itself.\
Both the CC algorithm used in our implementation and \gbbs{} is based on algorithms in an existing paper~\cite{dhulipala2020connectit}.
As mentioned, based on the results in \cite{dhulipala2020connectit}, the ``best'' CC algorithm can be very different for different types of graphs.
One can also plug in any CC algorithms to \ouralgo{} or \gbbs{} BCC algorithm to achieve better performance for specific input graphs.

In the \rootst{} step (\emph{generate rooted spanning trees}, the red bar), \ouralgo is significantly faster than \gbbs{}.
\gbbs{} is based on a BFS tree, and after computing the CCs of input graph $G$, it has to run BFS on $G$ again, which results in $O(m+n)$ work and $O(\textsc{Diam}(G)\log n)$ span.
In comparison, \ouralgo obtains the spanning trees from the \firstcc{} step, and only uses ETT in the \rootst{} step with $O(n)$ expected work and $O(\log n)$ span \whp.
As shown in \cref{fig:BCC_Total_Break}, this step for \gbbs{} is the dominating cost for \largediam{} graphs, and this is likely the case for other parallel BCC algorithms using BFS-based skeletons.
\ouralgo almost entirely saves the cost in this step (13$\times$ faster on average on \largediam{} graphs).
For \lowdiam{} graphs, the two algorithms perform similarly---\ouralgo{} is about 1.1$\times$ faster in this step.

In the \lowhigh{} step (the green bars), both \ouralgo and \gbbs compute the tags such as $\low$ and $\high$.
Since \ouralgo uses an AST, the values of the arrays are computed using 1D range-minimum query (see \cref{sec:bcc-details}) with $O(\log n)$ span.
\gbbs computes them by a bottom-up traversal on the BFS tree, with $O(\textsc{Diam}(G)\log n)$ span.
Hence, on \largediam{} graphs, \gbbs{} also consumes much time on this step, and \ouralgo is 1.2--830$\times$ faster than \gbbs.
On \lowdiam{} graphs, \gbbs{} also gets sufficient parallelism, and the performance for both algorithms is similar.

In summary, on all graphs, 
\ouralgo{} is faster than \gbbs{} mainly due to the efficiency in the \rootst{} and \lowhigh{} step, and the reason is that
our algorithm has polylogarithmic span, while \gbbs{} relies on the BFS spanning tree and requires $O(\textsc{Diam}(G)\log n)$ span.

\subsection{The Tarjan-Vishkin Algorithm}
\label{sec:exp:tv}
Although engineering Tarjan-Vishkin (TV)~\cite{tarjan1985efficient} is not the focus of this paper, for completeness, we also implemented a faithful TV algorithm.
\ifconference{Due to space limit, we give more details and report the numbers in the full paper, and summarize our findings here.}\iffullversion{We report the relative space usage of \ouralgo{}, TV, and \gbbs{} in \cref{fig:space_compare}, normalized to the most space-efficient implementation.}
Due to space inefficiency, our TV implementation cannot run on the three largest graphs (CW, HL14, and HL12) on our machines with 1.5TB memory.
We note that the smallest among them (CW) only takes about 300GB to store the graph, and our algorithm uses 572GB memory to process it.
On all graphs, TV uses 1.2--10.8$\times$ more space than \ouralgo{}.
\gbbs{} is about 20\% more space-efficient than \ouralgo{}. The reason is that they need to compute fewer number
of tags than \ouralgo{}.

\iffullversion{Regarding running time, we take the average of the running times for each algorithm on each category of the graph instances, and normalize to \ouralgo{}.
\begin{table}[!h]
    \small
    \centering
      \begin{tabular}{c|cccc}
             & \textbf{Ours} & \textbf{GBBS} &\textbf{Our-TV} & \textbf{SEQ} \\
      \hline
    Social & 1     & 1.67  & 10.1 & 21.2 \\
    Web   & 1     & 1.69  & 7.61  & 16.3 \\
    Road  & 1     & 9.89  & 1.94  & 7.18 \\
    $k$-NN  & 1     & 3.58  & 4.13  & 8.68 \\
    Synthetic & 1     & 14.9 & 3.96  & 17.0 \\
    \end{tabular}
    \label{tab:tv-overview}%
    \vspace{-.5em}
  \end{table}%
}\ifconference{Regarding running time, we report the running time of TV on all graphs in the full paper, and summarize the results here.} Due to the cost of explicitly constructing the \sketch{}, TV performs slowly on small-diameter graphs, and is slower than \gbbs even on $k$-NN graphs.
On all these graphs, the speedup for TV on 96 cores over \seq{} is only 1.4--3$\times$.
This is consistent with the findings in prior papers~\cite{cong2005experimental,slota2014simple}.
TV works well on road and synthetic graphs due to small edge-to-vertex ratio, so the $O(m)$ work and space for generating the skeleton does not dominate the running time.
In this case, polylogarithmic span allows TV to perform consistently better than \gbbs.
TV is faster than \seq{} on 96 cores on all graphs, but slower than \ouralgo.

\section{Conclusion}
In this paper, we propose the \ouralgo{} (Fencing on Arbitrary Spanning Tree) algorithm for parallel biconnectivity.
\ouralgo has $O(m+n)$ expected optimal work, polylogarithmic span (high parallelism), and uses $O(n)$ auxiliary space (space-efficient).
The theoretical efficiency also enables high performance.
On our machine with 96 cores and a variety of graph types,
\ouralgo{} outperforms all existing BCC implementations on all tested graphs.

\section*{Acknowledgement}
This work is supported by NSF grants CCF-2103483 and IIS-2227669, and UCR Regents Faculty Fellowships.
We thank anonymous reviewers for the useful feedbacks. 
\balance


\balance


\appendix

\iffullversion{
\section{More details and discussion for the Tarjan-Vishkin Algorithm}
\label{app:tv}
To parallelize BCC, the Tarjan-Vishkin algorithm~\cite{tarjan1985efficient} uses an arbitrary spanning tree (AST) $T$ instead of a DFS tree.
The spanning tree can be obtained by any parallel CC algorithm.
The first is the famous Euler tour technique (ETT) that efficiently roots a tree (see \cref{sec:prelim}).
Our algorithm also uses ETT.
The second technique is to build the skeleton $G'$ based on an AST.
Unfortunately, $G'$ in Tarjan-Vishkin is very large, making the algorithm less practical.

Given the input graph $G=(V,E)$ and an AST $T$, $G'=(E,E')$ where $E'$ consists of $(e_1,e_2)$ ($e_1,e_2\in E$ are edges in $G$) iff. one of the following conditions hold:
\begin{itemize}[leftmargin=*,topsep=0pt, partopsep=0pt,itemsep=0pt,parsep=0pt]
    \item $e_1=(u, p(u))$, $e_2=(u, v)$ in $G\setminus T$, and $u,v\in V, \first[v] < \first[u]$.
    \item $e_1=(u, p(u))$, $e_2=(v, p(v))$, and $(u, v)$ is a cross edge in $G\setminus T$.
    \item $e_1=(u, v)$, where $v=p(u)$ is not the root in $T$,
    and $e_2=(v, p(v))$, and there exists a non-tree edge $(x, y)$ such that $x\in T_u$ and $y\notin T_v$.
\end{itemize}

The above relationships can be determined by using the four axillary arrays $\first[\cdot]$, $\last[\cdot]$, $\low[\cdot]$, and $\high[\cdot]$ as mentioned in \cref{sec:tarjan-vishkin}.

In fact, we can prove that \ouralgo is equivalent to Tarjan-Vishkin.
However, the analysis for Tarjan-Vishkin is also quite involved (we refer to J{\'a}J{\'a}'s textbook for a good reference~\cite{JaJa92}).
Hence, we give a standalone analysis for \ouralgo in \cref{sec:correctness}, since we feel that understanding the analysis of Tarjan-Vishkin (correctness and cost bounds) and the analysis in \cref{sec:correctness} is in a similar level of difficulty.

In addition, we believe that the five algorithms we described and tested experimentally (Hopcroft-Tarjan, Tarjan-Vishkin, \ouralgo, \gbbs, \hipc{}) are similar, \emph{once we put them in the skeleton-connectivity framework and explicitly specify what the skeleton graph $G'$ is in each algorithm}.
Thus, the skeleton-connectivity framework brings in a different angle to understand parallel BCC algorithms, and eventually helps us come up with \ouralgo that is simple and efficient.

\section{Additional Proofs}\label{sec:add-proofs}

The proofs here are not very complicated, and should have been shown previously.
We provide them here mainly for completeness since the proofs of other lemmas in \cref{sec:correctness} use them.

\subsection{Proof of \cref{lem:common}}

\begin{aproof}
  Assume to the contrary that two BCCs $C_1$ and $C_2$ share at least two common vertices. After removing an arbitrary vertex from $C_1 \cup C_2$, there is at least one common vertex remaining. WLOG, we assume
  $u$ is a remaining common vertex. Because $C_1$ and $C_2$ are BCCs and $u$ is in both of them, all the remaining vertices in $C_1 \cup C_2$ are connected to $u$,
  so they remain in the same CC as $u$.
  Therefore, $C_1\cup C_2$ is a BCC, which contradicts with that
  $C_1$ and $C_2$ are two BCCs.
\end{aproof}

\subsection{Proof of \cref{lem:cycle}}

\begin{aproof}
  We first rewrite the cycle starts and ends with $v_i$ as $v_0$--$v_1$--$v_2$--\ldots--$v_j$--\ldots--$v_k$ ($v_0=v_k=v_i$).
  All the other vertices on the cycle appear exactly once.
  Then, there exists at least two disjoint paths that connect $v_i$ and $v_j$: one is $v_0$--\ldots--$v_j$, the other one is $v_j$--\ldots--$v_k$. Removing any vertex other than $v_i$ and $v_j$  disconnects at most one of the two paths, while the other path still connects $v_i$ and $v_j$.
  Thus, after removing any vertex, all the remaining vertices on the cycle are still connected, so all vertices on the cycle are in the same BCC.
\end{aproof}

\subsection{Proof of \cref{lem:bcc-connected}}

\begin{aproof}

  Assume to the contrary that $C$ has at least two CCs in $T$. For each CC, we find the shallowest node. Let $u$ and $v$ be the shallowest two of these two CCs. WLOG assume $u$ is no deeper than $v$.
  There are two possible positions for $u$ and $v$ in the spanning tree $T$. We first show that in both cases, $v$'s parent $w$ is biconnected with $u$.

  Case 1: neither $u$ nor $v$ is the ancestor of the other.
  Based on our assumption, there exists at least one path $P$ from $u$ to $v$ using vertices in $C$.
  Note that no vertices on the tree path from $u$ to $v$ are included in $P$.
  We now show that there are two disjoint paths that connect $w$ and $u$: 1) the tree path between $w$ and $u$, which does not contain intermediate nodes from $C$, and 2) the path from $w$ to $v$ then to $u$, only using intermediate nodes from $C$.
  Hence, $w$ and $u$ are biconnected. 

  Case 2: $u$ is $v$'s ancestor. We can similarly show that there are at least two paths from $w$ to the nearest vertex in $u$'s component, one using the tree path while the other using vertices in $C$.
  Hence, $w$ and $u$ are biconnected.

  In both cases, we can show that $w$ and $u$ are biconnected. For any $x\in C$, removing any other vertex $y\in C$, $x$ and $w$ are still connected through either $u$ or $v$. Hence, $w$ and $C$ are biconnected, contradicting that $v$ is the shallowest tree node in the BCC ($w$ is $v$'s parent).
\end{aproof}

\subsection{Proof of \cref{lem:art-head}}

\begin{aproof}
  We first prove that each non-root \bcchead{} is an articulation point.
  A \bcchead{} $h$ is the shallowest node for a BCC $C$.
  Let the $c$ be one of $h$'s children in $C$.
  Assume to the contrary that $h$ is not an articulation point, then $G$ is still connected after removing $h$, including $p(h)$ and $c$.
  Removing any other vertex other than $c$, $h$, and $p(h)$ does not disconnect $c$ and $h$ based on the definition of the BCC, so $c$ and $p(h)$ is also connected using an additional tree edge $h$--$p(h)$.
  Combining both cases, $p(h)$ is biconnected with $c$, contradicting that $h$ is a BCC head.

  We now show an articulation point $a$ must be a BCC head.
  Assume to the contrary that $a$ is not a BCC head, then based on \cref{lem:bcc-connected}, all $a$'s children must be in the same BCC as $p(a)$.
  Since $a$ is an articulation point, removing it disconnects $G$.
  However, all $a$'s children's subtrees are still connected, so as $V\setminus T_a$ using tree edges.
  Hence, at least one of $a$'s children, referred to as $c$, is disconnected from $p(a)$, since otherwise $a$ is not an articulation point.
  In this case, $a$ is the BCC head of the BCC which $a$ and $c$ is in, contradicting the assumption.
\end{aproof}

\subsection{Proof of \cref{lem:function}}
\begin{aproof}
To prove that the skeleton is generated correctly, we show that the functions \textsc{InSkeleton}, \textsc{Fence}, and \textsc{Back} work as expected, and \textsc{InSkeleton} returns \emph{true} iff. edge $u$--$v$ is a plain edge or a cross edge.
We first prove that a vertex $u$ is $v$'s ancestor if and only if $\first[u]\le \first[v]$ and $\last[u]\ge \last[v]$, which is exactly \cref{line:bcc-backward}.

On an Euler tour of a spanning tree, each edge appears exactly twice (one time in each direction) in a DFS order. $\first[\cdot]$/$\last[\cdot]$ stores the time stamps each the vertex $\first$/$\last$ appears on the Euler tour.
We first show that $\forall v \in T_u$, $\first[u]\leq \first[v]$ and $\last[v] \leq \last[u]$.
This is because all the tree edges in $T_u$ are traversed after $u$ have been traversed, so $\forall v \in T_u, \first[v] \geq \first[u]$; and $u$ last appears when all the tree edges in $T_u$ have been traversed, so $\forall v \in T_u, \last[v] \leq \last[u]$.


We show that if $u$ is an ancestor of $v$, then the function \cref{line:bcc-backward} in \cref{alg:bcc} returns true.
If $u$ is an ancestor of $v$, then $v \in T_u$, so that $\first[u] \leq \first[v]$ and $\last[u] \geq \last[v] \geq \first[v]$.
Then we will show if $u$ is not an ancestor of $v$, then the function \cref{line:bcc-backward} in \cref{alg:bcc} returns false (at least one of the two conditions is false).
If $u$ is not an ancestor of $v$, there are two cases for $u$: $u \in T_v$ or $u\notin T_v$.
If $u\in T_v$, then $\first[v] \leq \first[u]$, so the first condition in function \cref{line:bcc-backward} is false. If $u \notin T_v$, either $\last[u] < \first[v]$ and $\last[v] < \first[u]$. If $\last[u]<\first[v]$, the second condition in function \cref{line:bcc-backward} is false. If $\last[v] < \first[u]$, because $\first[v]<\last[v]<\first[u]$, the first condition is false.
Therefore, $u$ is an ancestor of $v$ iff. the function \cref{line:bcc-backward} in \cref{alg:bcc} returns true.
This function is called on edge $u$--$v$  by the function \textsc{InSkeleton} only when it is a tree edge. Therefore, $u$ is a non-parent ancestor of $v$.
Therefore, \textsc{InSkeleton} returns \emph{true} on \cref{line:bcc-back} iff. $u$--$v$ is a cross edge.

We then prove that the function \cref{line:bcc-critical} in \cref{alg:bcc} can correctly determine whether a tree edge $u$--$v$ is a fence edge.

We first show that if tree edge $u$--$v$ is a fence edge, \cref{line:bcc-critical} in \cref{alg:bcc} returns \emph{true}.
We just showed that $\forall x \in T_u, \first[u] \leq \first[x]$ and $\last[u] \geq \last[x] \geq \first[x]$. If tree edge $u$--$v$  is a fence edge, then
for all the edges with one endpoint in $T_v$, the other endpoint must be in $T_u$.
Recall that $\low[v]$ is the earliest (with the smallest $\first$ value) vertex connected to $v$'s subtree.
This means that this earliest vertex is also in $T_u$, and therefore $\low[v]\ge \first[u]$.
Similarly, $\high[v]$, which is the latest (with the largest $\first$ value) vertex connected to $v$'s subtree, should also be in $T_u$.
Therefore $\last[u]\ge \high[v]$.

Then we show that if tree edge $u$--$v$ is not a fence edge, \cref{line:bcc-critical} in \cref{alg:bcc} returns false. If tree edge $u$--$v$ is not a fence edge, then there exists an edge $x'$--$y'$, where $x'\in T_v$ and $y' \notin T_u$.  Because $y' \notin T_u$, either $\first[y'] < \first[u]$ or $\first[y'] > \last[y']$.
If $\first[y'] < \first[u]$, then

\begin{equation*}
\low[v] \leq w_1[x'] \leq \first[y']<\first[u]
\vspace{.2in}
\end{equation*}

Then the first condition in \cref{line:bcc-critical} in \cref{alg:bcc} is false.
If $\first[y'] > \last[u]$, then

\begin{equation*}
  \high[v] \geq w_2[x'] \geq \first[y']>\last[u]
  \vspace{.2in}
\end{equation*}

Then the second condition in \cref{line:bcc-critical} in \cref{alg:bcc} is false.
Therefore, the tree edge $u$--$v$ is a fence edge iff. function \cref{line:bcc-critical} in \cref{alg:bcc} returns \emph{true}.

In summary, \textsc{InSkeleton} returns \emph{true} iff. edge $u$--$v$ is a plain edge or a cross edge. Therefore, the skeleton $G'$ can be determined correctly by \textsc{InSkeleton}.
\end{aproof}

\subsection{Proof of \cref{thm:connect}}
\label{app:cc}
\vspace{.5em}
\begin{aproof}
Before we show the analysis, we will first review the two key techniques that name this algorithm in \connectit{}: low-diameter decomposition (LDD)~\cite{miller2013parallel} and a union-find structure by Jayanti et al.~\cite{jayanti2019randomized}.
A $(\beta, d)$-decomposition of a graph $G=(V,E)$ is a partition of $V$ into subsets $V_1, V_2, \cdots, V_k$ such that
(1) the diameter of each $V_i$ is at most~$d$, and
(2) the number of edges $(u, v) \in E$ with endpoints in different subsets, i.e., such that $u \in V_i, v \in V_j$ and $i \neq j$, is at most $\beta m$.
A parallel $(\beta, O({(\log n)}/{\beta}))$ decomposition algorithm is provided by Miller et. al.~\cite{miller2013parallel}, using $O(n+m)$ work and $O((\log^2 n)/\beta)$ span \whp.
The high-level idea of LDD is to start with a single source and search out using BFS.
Then in later rounds, we exponentially add new sources to the frontier and continue BFS processes.
By controlling the speed to add new sources, the entire BFS will finish in $O((\log n)/\beta)$ rounds, leaving at most $\beta m$ edges with endpoints from different sources.

Once the LDD is computed, the algorithm will examine all cross edges (endpoints from different sources) using a union-find structure by Jayanti et al.~\cite{jayanti2019randomized} to merge different components.
The algorithm either performs finds naively without using any path compression or uses a strategy called Find-Two-Try-Split.
Such strategies guarantee provably-efficient bounds.
The original bound is $O(l\cdot(\alpha(n,l/(np)) + \log(np/l + 1)))$ expected work and $O(\log n)$ PRAM time for a problem instance with $l$ operations on $n$ elements on a PRAM with $p$ processors.
When translating this bound to the binary fork-join model, all $l$ operations can be in parallel in the worst case, which leads to the work bound as $O(l\log n)$.

We note that if we set $\beta=1/\log n$, the LDD takes $O(n+m)$ work and $O((\log^3 n)/\beta)$ span \whp, and the union-find part takes $O(\beta m \log n)=O(m)$ work and $O(\log^2 n)$ span.
Combining the two pieces together gives $O(n+m)$ work and $O((\log^3 n)/\beta)$ span \whp for the ``LDD-UF-JTB'' algorithm.
\end{aproof}

\begin{figure*}[t]
  \centering
  \includegraphics[width=0.8\textwidth]{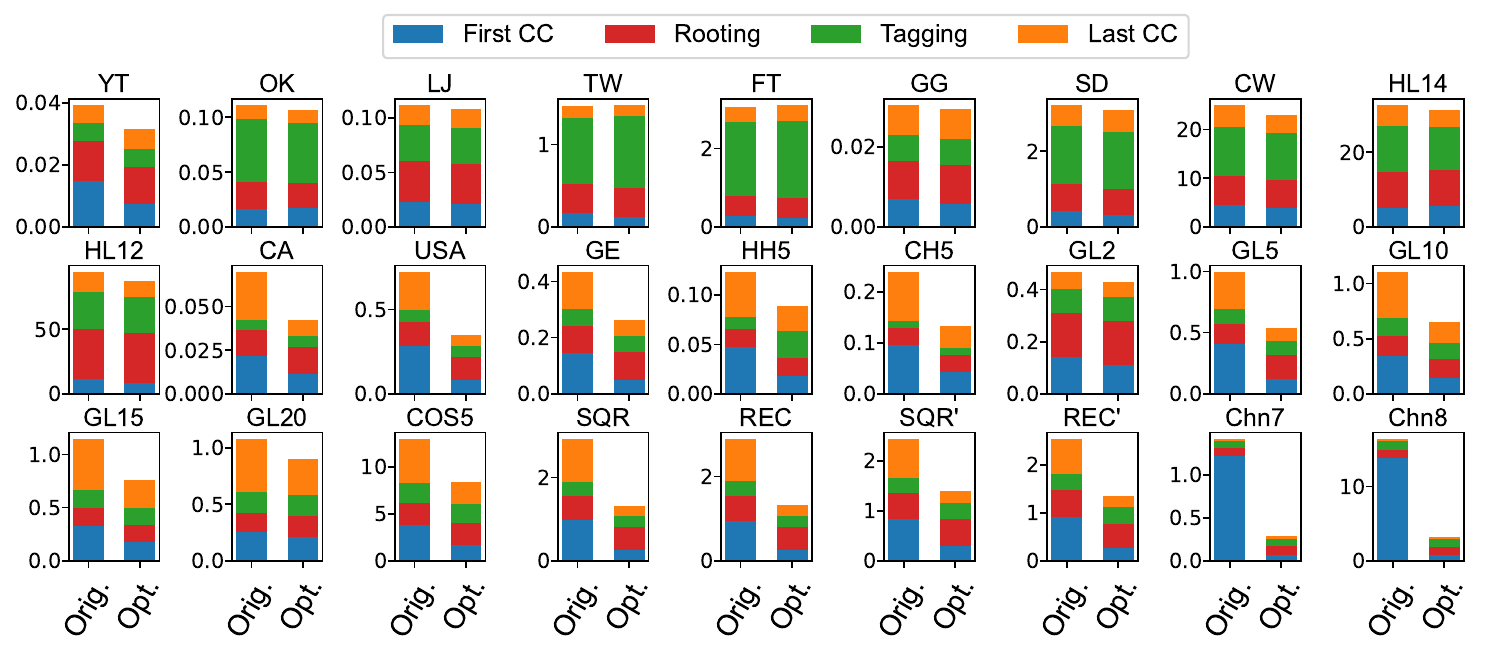}
  \caption{\small \textbf{Optimized BCC breakdown.} $y$-axis is the running time in seconds. "Orig."$=$ our original BCC implementation, "Opt."$=$ our implementation optimized with hash bags and local search proposed in~\cite{scc2022}.}\label{fig:BCC_Break_loc}
\end{figure*} 

%



\section{Performance Analysis of the Local Search Optimality}
In \Ouralgo, CC is an important primitive used in \firstcc{}  and \lastcc{}.
We optimized the CC implementation using hash bags and local searches proposed by a recent paper~\cite{scc2022}.
These optimizations function as a parallel granularity control.
When the frontier (vertices being processed in one round) is small,
the algorithm explores multi-hop neighbors of the frontier instead of one-hop neighbors to saturate all threads with sufficient work.
It helps reduce the number of total rounds in a connectivity search,
thus reducing the synchronization costs between rounds.
It works favorably well on \largediam{} graphs.
We measure the improvement from the optimizations in \cref{fig:BCC_Break_loc}, where \emph{Orig.} is the version without the optimizations and \emph{Opt.} is the version with the optimizations.
As shown in \cref{fig:BCC_Break_loc},
on \lowdiam{} graphs, \emph{Orig.} and \emph{Opt.} have similar performance.
On \largediam{}, \emph{Opt.} can be 1.1--4.5$\times$ faster than \emph{Orig.} and is 1.8$\times$ faster on average. 

\section{Performance of the Tarjan-Vishkin Algorithm and Space Usage}

For completeness, we also implemented the faithful Tarjan-Vishkin algorithm~\cite{tarjan1985efficient} discussed in \cref{app:tv}.
We acknowledge that some existing papers~\cite{cong2005experimental,edwards2012better} discussed some possible optimizations for Tarjan-Vishkin.
However, engineering the Tarjan-Vishkin algorithm is not the focus of this paper, and we mainly use it to measure the space usage and get a sense on how Tarjan-Vishkin compares to other existing BCC algorithms.
Our conclusions are consistent with the results drawn from the previous papers~\cite{cong2005experimental,slota2014simple}.

\begin{table}[htbp]
    \small
    \centering
      \begin{tabular}{cc|cccc}
            &       & \multirow{2}[1]{*}{\textbf{Ours}} & \multirow{2}[1]{*}{\textbf{GBBS}} & \multirow{2}[1]{*}{\textbf{TV}} & \multirow{2}[1]{*}{\textbf{SEQ}} \\
            &       &       &       &       &  \\
      \midrule
      \multirow{5}[2]{*}{\begin{sideways}\textbf{Social}\end{sideways}} & \textbf{YT} & \textbf{0.030} & 0.040 & 0.076 & 0.175 \\
            & \textbf{OK} & 0.103 & 0.158 & 2.10  & 3.14 \\
            & \textbf{LJ} & 0.104 & 0.159 & 0.859 & 1.87 \\
            & \textbf{TW} & 1.44  & 2.83  & 25.7  & 49.2 \\
            & \textbf{FT} & 3.10  & 6.44  & 41.6  & 122 \\
      \midrule
      \multirow{5}[2]{*}{\begin{sideways}\textbf{Web}\end{sideways}} & \textbf{GG} & 0.029 & 0.045 & 0.119 & 0.255 \\
            & \textbf{SD} & 3.11  & 5.61  & 43.0  & 92.3 \\
            & \textbf{CW} & 22.9  & 39.7  & N/A   & 695 \\
            & \textbf{HL14} & 31.1  & 50.7  & N/A   & 1011 \\
            & \textbf{HL12} & 89.1  & 105   & N/A   & 3027 \\
      \midrule
      \multirow{3}[2]{*}{\begin{sideways}\textbf{Road}\end{sideways}} & \textbf{CA} & 0.040 & 0.372 & 0.079 & 0.206 \\
            & \textbf{USA} & 0.336 & 4.64  & 0.673 & 2.25 \\
            & \textbf{GE} & 0.267 & 2.02  & 0.492 & 2.88 \\
      \midrule
      \multirow{8}[2]{*}{\begin{sideways}\boldmath{}\textbf{$\boldsymbol{k}$-NN}\unboldmath{}\end{sideways}} & \textbf{HH5} & 0.073 & 0.447 & 0.169 & 0.509 \\
            & \textbf{CH5} & 0.128 & 1.44  & 0.380 & 0.528 \\
            & \textbf{GL2} & 0.402 & 1.53  & 0.771 & 2.51 \\
            & \textbf{GL5} & 0.472 & 2.80  & 1.73  & 4.03 \\
            & \textbf{GL10} & 0.668 & 1.64  & 3.68  & 7.07 \\
            & \textbf{GL15} & 0.751 & 1.51  & 5.06  & 8.92 \\
            & \textbf{GL20} & 0.861 & 1.48  & 6.91  & 10.2 \\
            & \textbf{COS5} & 8.46  & 17.5  & 50.1  & 120 \\
      \midrule
      \multirow{6}[2]{*}{\begin{sideways}\textbf{Synthetic}\end{sideways}} & \textbf{SQR} & 1.32  & 15.4  & 6.79  & 24.4 \\
            & \textbf{REC} & 1.35  & 47.0  & 6.99  & 16.8 \\
            & \textbf{SQR'} & 1.31  & 12.5  & 2.76  & 10.6 \\
            & \textbf{REC'} & 1.37  & 22.4  & 2.83  & 10.7 \\
            & \textbf{Chn7} & 0.278 & 81.6  & 0.343 & 3.33 \\
            & \textbf{Chn8} & 3.25  & 957   & 3.97  & 38.9 \\
      \bottomrule
      \end{tabular}%

    \caption{\small\textbf{Tarjan-Vishkin running time (in seconds).}
    ``N/A'' $=$ not applicable because of out of memory.
    }
    \label{tab:tv-full}%
  \end{table}%
 
\begin{figure*}[t]
    \centering
    \includegraphics[width=\textwidth]{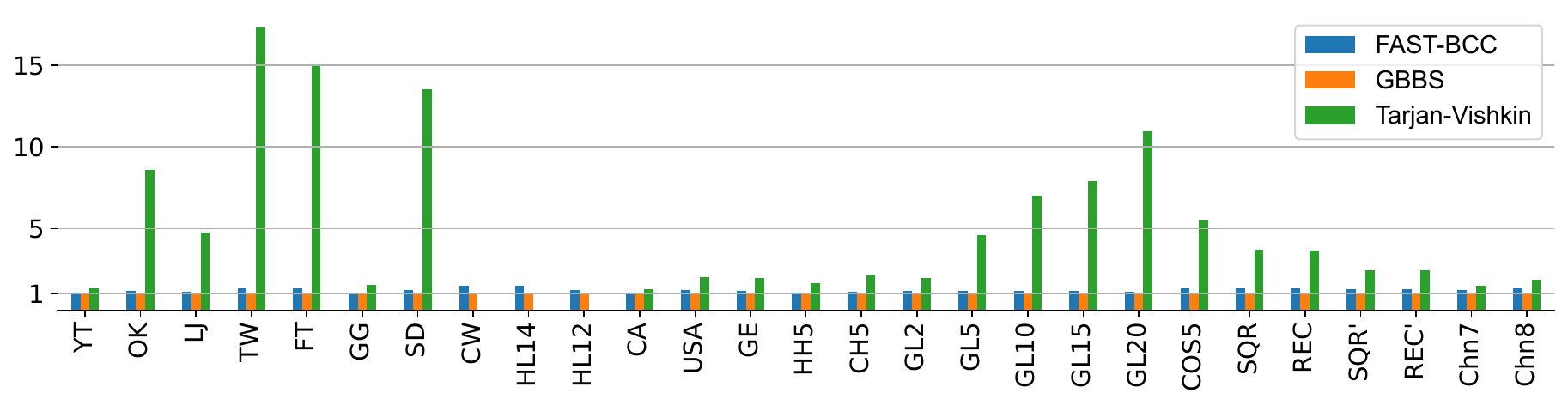}
    \caption{\small \textbf{Space Usage Comparision.} $y$-axis is the space usage normalized to GBBS. Lower is better.}\label{fig:space_compare}
\end{figure*} 

The running time and space usage are given in \cref{tab:tv-full} and \cref{fig:space_compare}.
For our implementations, Tarjan-Vishkin and use up to 11$\times$ extra space than \ouralgo on FT and SD (including the space to store the input graph).
The space overhead is decided by edge-to-vertex ratio, and for graphs with smaller ratios (e.g., chain graphs), the overhead is small.
Our TV implementation cannot run on the three largest graphs (CW, HL14, and HL12) on our machines with 1.5TB memory.
We note that the smallest among them (CW) only takes about 300GB to store the graph, and \ouralgo uses 572GB memory to process it.
Since all three graphs have relatively large edge-to-vertex ratios, our TV implementations are unlikely to execute on these graphs for shared-memory machines in foreseeable future.
\gbbs{} is slightly more space-efficient than \ouralgo{}, and takes about 20\% less space than us. The reason is that they need to compute fewer number
of tags than \ouralgo{}.

On all graphs, TV is faster than SEQ on 96 cores, but slower than \ouralgo.
The overhead of TV is due to the cost to explicitly construct the skeleton.
On social, web, and $k$-NN graphs, the speedup for TV on 96 cores over \seq{} is only 1.4--3$\times$.
TV works well on road and synthetic graphs due to small edge-to-vertex ratio, so the $O(m)$ work and space for generating the skeleton does not dominate the running time.
In this case, polylogarithmic span allows TV to perform consistently better than \gbbs.

}


\begin{thebibliography}{73}


\ifx \showCODEN    \undefined \def \showCODEN     #1{\unskip}     \fi
\ifx \showDOI      \undefined \def \showDOI       #1{#1}\fi
\ifx \showISBNx    \undefined \def \showISBNx     #1{\unskip}     \fi
\ifx \showISBNxiii \undefined \def \showISBNxiii  #1{\unskip}     \fi
\ifx \showISSN     \undefined \def \showISSN      #1{\unskip}     \fi
\ifx \showLCCN     \undefined \def \showLCCN      #1{\unskip}     \fi
\ifx \shownote     \undefined \def \shownote      #1{#1}          \fi
\ifx \showarticletitle \undefined \def \showarticletitle #1{#1}   \fi
\ifx \showURL      \undefined \def \showURL       {\relax}        \fi
\providecommand\bibfield[2]{#2}
\providecommand\bibinfo[2]{#2}
\providecommand\natexlab[1]{#1}
\providecommand\showeprint[2][]{arXiv:#2}

\bibitem[Aggarwal et~al\mbox{.}(1988)]%
        {aggarwal1988parallel}
\bibfield{author}{\bibinfo{person}{Alok Aggarwal}, \bibinfo{person}{Bernard Chazelle}, \bibinfo{person}{Leo Guibas}, \bibinfo{person}{Colm {\'O}'D{\'u }nlaing}, {and} \bibinfo{person}{Chee Yap}.} \bibinfo{year}{1988}\natexlab{}.
\newblock \showarticletitle{Parallel computational geometry}.
\newblock \bibinfo{journal}{\emph{Algorithmica}} \bibinfo{volume}{3}, \bibinfo{number}{1} (\bibinfo{year}{1988}), \bibinfo{pages}{293--327}.
\newblock


\bibitem[Agrawal et~al\mbox{.}(2014)]%
        {agrawal2014batching}
\bibfield{author}{\bibinfo{person}{Kunal Agrawal}, \bibinfo{person}{Jeremy~T. Fineman}, \bibinfo{person}{Kefu Lu}, \bibinfo{person}{Brendan Sheridan}, \bibinfo{person}{Jim Sukha}, {and} \bibinfo{person}{Robert Utterback}.} \bibinfo{year}{2014}\natexlab{}.
\newblock \showarticletitle{Provably Good Scheduling for Parallel Programs That Use Data Structures Through Implicit Batching}. In \bibinfo{booktitle}{\emph{{ACM} Symposium on Parallelism in Algorithms and Architectures (SPAA)}}.
\newblock


\bibitem[Anderson et~al\mbox{.}(2022)]%
        {anderson2022problem}
\bibfield{author}{\bibinfo{person}{Daniel Anderson}, \bibinfo{person}{Guy~E Blelloch}, \bibinfo{person}{Laxman Dhulipala}, \bibinfo{person}{Magdalen Dobson}, {and} \bibinfo{person}{Yihan Sun}.} \bibinfo{year}{2022}\natexlab{}.
\newblock \showarticletitle{The problem-based benchmark suite (PBBS), V2}. In \bibinfo{booktitle}{\emph{{ACM} Symposium on Principles and Practice of Parallel Programming (PPOPP)}}. \bibinfo{pages}{445--447}.
\newblock


\bibitem[Arge et~al\mbox{.}(2002)]%
        {arge2002cache}
\bibfield{author}{\bibinfo{person}{Lars Arge}, \bibinfo{person}{Michael Bender}, \bibinfo{person}{Erik Demaine}, \bibinfo{person}{Bryan Holland-Minkley}, {and} \bibinfo{person}{Ian Munro}.} \bibinfo{year}{2002}\natexlab{}.
\newblock \showarticletitle{Cache-oblivious priority queue and graph algorithm applications}. In \bibinfo{booktitle}{\emph{{ACM} Symposium on Theory of Computing (STOC)}}. \bibinfo{pages}{268--276}.
\newblock


\bibitem[Ausiello et~al\mbox{.}(2011)]%
        {ausiello2011real}
\bibfield{author}{\bibinfo{person}{Giorgio Ausiello}, \bibinfo{person}{Donatella Firmani}, {and} \bibinfo{person}{Luigi Laura}.} \bibinfo{year}{2011}\natexlab{}.
\newblock \showarticletitle{Real-time anomalies detection and analysis of network structure, with application to the Autonomous System network}. In \bibinfo{booktitle}{\emph{International Wireless Communications and Mobile Computing Conference}}. IEEE, \bibinfo{pages}{1575--1579}.
\newblock


\bibitem[Bachmaier et~al\mbox{.}(2005)]%
        {bachmaier2005radial}
\bibfield{author}{\bibinfo{person}{Christian Bachmaier}, \bibinfo{person}{Franz~J Brandenburg}, {and} \bibinfo{person}{Michael Forster}.} \bibinfo{year}{2005}\natexlab{}.
\newblock \showarticletitle{Radial level planarity testing and embedding in linear time}. In \bibinfo{booktitle}{\emph{J. Graph Algorithms and Applications}}. Citeseer.
\newblock


\bibitem[Backstrom et~al\mbox{.}(2006)]%
        {backstrom2006group}
\bibfield{author}{\bibinfo{person}{Lars Backstrom}, \bibinfo{person}{Dan Huttenlocher}, \bibinfo{person}{Jon Kleinberg}, {and} \bibinfo{person}{Xiangyang Lan}.} \bibinfo{year}{2006}\natexlab{}.
\newblock \showarticletitle{Group formation in large social networks: membership, growth, and evolution}. In \bibinfo{booktitle}{\emph{ACM International Conference on Knowledge Discovery and Data Mining (SIGKDD)}}. \bibinfo{pages}{44--54}.
\newblock


\bibitem[Ben-David et~al\mbox{.}(2016)]%
        {BBFGGMS16}
\bibfield{author}{\bibinfo{person}{Naama Ben-David}, \bibinfo{person}{Guy~E. Blelloch}, \bibinfo{person}{Jeremy~T. Fineman}, \bibinfo{person}{Phillip~B. Gibbons}, \bibinfo{person}{Yan Gu}, \bibinfo{person}{Charles McGuffey}, {and} \bibinfo{person}{Julian Shun}.} \bibinfo{year}{2016}\natexlab{}.
\newblock \showarticletitle{Parallel Algorithms for Asymmetric Read-Write Costs}. In \bibinfo{booktitle}{\emph{{ACM} Symposium on Parallelism in Algorithms and Architectures (SPAA)}}.
\newblock


\bibitem[Ben-David et~al\mbox{.}(2018)]%
        {BBFGGMS18}
\bibfield{author}{\bibinfo{person}{Naama Ben-David}, \bibinfo{person}{Guy~E. Blelloch}, \bibinfo{person}{Jeremy~T Fineman}, \bibinfo{person}{Phillip~B Gibbons}, \bibinfo{person}{Yan Gu}, \bibinfo{person}{Charles McGuffey}, {and} \bibinfo{person}{Julian Shun}.} \bibinfo{year}{2018}\natexlab{}.
\newblock \showarticletitle{Implicit Decomposition for Write-Efficient Connectivity Algorithms}. In \bibinfo{booktitle}{\emph{{IEEE} International Parallel and Distributed Processing Symposium (IPDPS)}}.
\newblock


\bibitem[Blelloch et~al\mbox{.}(2020a)]%
        {blelloch2020parlaylib}
\bibfield{author}{\bibinfo{person}{Guy~E. Blelloch}, \bibinfo{person}{Daniel Anderson}, {and} \bibinfo{person}{Laxman Dhulipala}.} \bibinfo{year}{2020}\natexlab{a}.
\newblock \showarticletitle{ParlayLib --- a toolkit for parallel algorithms on shared-memory multicore machines}. In \bibinfo{booktitle}{\emph{{ACM} Symposium on Parallelism in Algorithms and Architectures (SPAA)}}. \bibinfo{pages}{507--509}.
\newblock


\bibitem[Blelloch et~al\mbox{.}(2008)]%
        {BCGRCK08}
\bibfield{author}{\bibinfo{person}{Guy~E. Blelloch}, \bibinfo{person}{Rezaul~Alam Chowdhury}, \bibinfo{person}{Phillip~B. Gibbons}, \bibinfo{person}{Vijaya Ramachandran}, \bibinfo{person}{Shimin Chen}, {and} \bibinfo{person}{Michael Kozuch}.} \bibinfo{year}{2008}\natexlab{}.
\newblock \showarticletitle{Provably good multicore cache performance for divide-and-conquer algorithms}. In \bibinfo{booktitle}{\emph{{ACM-SIAM} Symposium on Discrete Algorithms (SODA)}}.
\newblock


\bibitem[Blelloch et~al\mbox{.}(2011)]%
        {BlellochFiGi11}
\bibfield{author}{\bibinfo{person}{Guy~E. Blelloch}, \bibinfo{person}{Jeremy~T. Fineman}, \bibinfo{person}{Phillip~B. Gibbons}, {and} \bibinfo{person}{Harsha~Vardhan Simhadri}.} \bibinfo{year}{2011}\natexlab{}.
\newblock \showarticletitle{Scheduling Irregular Parallel Computations on Hierarchical Caches}. In \bibinfo{booktitle}{\emph{{ACM} Symposium on Parallelism in Algorithms and Architectures (SPAA)}}. \bibinfo{pages}{355--366}.
\newblock


\bibitem[Blelloch et~al\mbox{.}(2020b)]%
        {blelloch2020optimal}
\bibfield{author}{\bibinfo{person}{Guy~E. Blelloch}, \bibinfo{person}{Jeremy~T. Fineman}, \bibinfo{person}{Yan Gu}, {and} \bibinfo{person}{Yihan Sun}.} \bibinfo{year}{2020}\natexlab{b}.
\newblock \showarticletitle{Optimal parallel algorithms in the binary-forking model}. In \bibinfo{booktitle}{\emph{{ACM} Symposium on Parallelism in Algorithms and Architectures (SPAA)}}. \bibinfo{pages}{89--102}.
\newblock


\bibitem[Blelloch and Gibbons(2004)]%
        {BG04}
\bibfield{author}{\bibinfo{person}{Guy~E. Blelloch} {and} \bibinfo{person}{Phillip~B. Gibbons}.} \bibinfo{year}{2004}\natexlab{}.
\newblock \showarticletitle{Effectively sharing a cache among threads}. In \bibinfo{booktitle}{\emph{{ACM} Symposium on Parallelism in Algorithms and Architectures (SPAA)}}.
\newblock


\bibitem[Blelloch et~al\mbox{.}(2010)]%
        {blelloch2010low}
\bibfield{author}{\bibinfo{person}{Guy~E. Blelloch}, \bibinfo{person}{Phillip~B. Gibbons}, {and} \bibinfo{person}{Harsha~Vardhan Simhadri}.} \bibinfo{year}{2010}\natexlab{}.
\newblock \showarticletitle{Low depth cache-oblivious algorithms}. In \bibinfo{booktitle}{\emph{{ACM} Symposium on Parallelism in Algorithms and Architectures (SPAA)}}.
\newblock


\bibitem[Blelloch et~al\mbox{.}(2018)]%
        {blelloch2018geometry}
\bibfield{author}{\bibinfo{person}{Guy~E. Blelloch}, \bibinfo{person}{Yan Gu}, \bibinfo{person}{Julian Shun}, {and} \bibinfo{person}{Yihan Sun}.} \bibinfo{year}{2018}\natexlab{}.
\newblock \showarticletitle{Parallel Write-Efficient Algorithms and Data Structures for Computational Geometry}. In \bibinfo{booktitle}{\emph{{ACM} Symposium on Parallelism in Algorithms and Architectures (SPAA)}}.
\newblock


\bibitem[Blelloch et~al\mbox{.}(2020c)]%
        {blelloch2020randomized}
\bibfield{author}{\bibinfo{person}{Guy~E. Blelloch}, \bibinfo{person}{Yan Gu}, \bibinfo{person}{Julian Shun}, {and} \bibinfo{person}{Yihan Sun}.} \bibinfo{year}{2020}\natexlab{c}.
\newblock \showarticletitle{Randomized Incremental Convex Hull is Highly Parallel}. In \bibinfo{booktitle}{\emph{{ACM} Symposium on Parallelism in Algorithms and Architectures (SPAA)}}.
\newblock


\bibitem[Blelloch and Reid-Miller(1998)]%
        {Blelloch1998}
\bibfield{author}{\bibinfo{person}{Guy~E. Blelloch} {and} \bibinfo{person}{Margaret Reid-Miller}.} \bibinfo{year}{1998}\natexlab{}.
\newblock \showarticletitle{Fast Set Operations Using Treaps}. In \bibinfo{booktitle}{\emph{{ACM} Symposium on Parallelism in Algorithms and Architectures (SPAA)}}. \bibinfo{pages}{16--26}.
\newblock


\bibitem[Blelloch and Reid-Miller(1999)]%
        {blelloch1999pipelining}
\bibfield{author}{\bibinfo{person}{Guy~E. Blelloch} {and} \bibinfo{person}{Margaret Reid-Miller}.} \bibinfo{year}{1999}\natexlab{}.
\newblock \showarticletitle{Pipelining with futures}.
\newblock \bibinfo{journal}{\emph{Theory of Computing Systems (TOCS)}} \bibinfo{volume}{32}, \bibinfo{number}{3} (\bibinfo{year}{1999}), \bibinfo{pages}{213--239}.
\newblock


\bibitem[Blelloch et~al\mbox{.}(2012)]%
        {BST12}
\bibfield{author}{\bibinfo{person}{Guy~E. Blelloch}, \bibinfo{person}{Harsha~Vardhan Simhadri}, {and} \bibinfo{person}{Kanat Tangwongsan}.} \bibinfo{year}{2012}\natexlab{}.
\newblock \showarticletitle{Parallel and {I/O} efficient set covering algorithms}. In \bibinfo{booktitle}{\emph{{ACM} Symposium on Parallelism in Algorithms and Architectures (SPAA)}}.
\newblock


\bibitem[Blumofe and Leiserson(1998)]%
        {BL98}
\bibfield{author}{\bibinfo{person}{Robert~D. Blumofe} {and} \bibinfo{person}{Charles~E. Leiserson}.} \bibinfo{year}{1998}\natexlab{}.
\newblock \showarticletitle{Space-Efficient Scheduling of Multithreaded Computations}.
\newblock \bibinfo{journal}{\emph{{SIAM} J. on Computing}} \bibinfo{volume}{27}, \bibinfo{number}{1} (\bibinfo{year}{1998}).
\newblock


\bibitem[Boyer and Myrvold(2006)]%
        {boyer2006simplified}
\bibfield{author}{\bibinfo{person}{John~M Boyer} {and} \bibinfo{person}{Wendy~J Myrvold}.} \bibinfo{year}{2006}\natexlab{}.
\newblock \showarticletitle{Simplified $o(n)$ planarity by edge addition}.
\newblock \bibinfo{journal}{\emph{J. Graph Algorithms and Applications}}  \bibinfo{volume}{5} (\bibinfo{year}{2006}), \bibinfo{pages}{241}.
\newblock


\bibitem[Chaitanya and Kothapalli(2015)]%
        {chaitanya2015simple}
\bibfield{author}{\bibinfo{person}{Meher Chaitanya} {and} \bibinfo{person}{Kishore Kothapalli}.} \bibinfo{year}{2015}\natexlab{}.
\newblock \showarticletitle{A simple parallel algorithm for biconnected components in sparse graphs}. In \bibinfo{booktitle}{\emph{International Parallel and Distributed Processing Symposium (IPDPS) Workshop}}. IEEE, \bibinfo{pages}{395--404}.
\newblock


\bibitem[Chaitanya and Kothapalli(2016)]%
        {chaitanya2016efficient}
\bibfield{author}{\bibinfo{person}{Meher Chaitanya} {and} \bibinfo{person}{Kishore Kothapalli}.} \bibinfo{year}{2016}\natexlab{}.
\newblock \showarticletitle{Efficient multicore algorithms for identifying biconnected components}.
\newblock \bibinfo{journal}{\emph{International Journal of Networking and Computing}} \bibinfo{volume}{6}, \bibinfo{number}{1} (\bibinfo{year}{2016}), \bibinfo{pages}{87--106}.
\newblock


\bibitem[Cheriyan and Thurimella(1991)]%
        {cheriyan1991algorithms}
\bibfield{author}{\bibinfo{person}{Joseph Cheriyan} {and} \bibinfo{person}{Ramakrishna Thurimella}.} \bibinfo{year}{1991}\natexlab{}.
\newblock \showarticletitle{Algorithms for parallel k-vertex connectivity and sparse certificates}. In \bibinfo{booktitle}{\emph{{ACM} Symposium on Theory of Computing (STOC)}}. \bibinfo{pages}{391--401}.
\newblock


\bibitem[Chiang et~al\mbox{.}(1995)]%
        {chiang1995external}
\bibfield{author}{\bibinfo{person}{Yi-Jen Chiang}, \bibinfo{person}{Michael Goodrich}, \bibinfo{person}{Edward Grove}, \bibinfo{person}{Roberto Tamassia}, \bibinfo{person}{Darren~Erik Vengroff}, {and} \bibinfo{person}{Jeffrey Vitter}.} \bibinfo{year}{1995}\natexlab{}.
\newblock \showarticletitle{External-Memory Graph Algorithms.}. In \bibinfo{booktitle}{\emph{{ACM-SIAM} Symposium on Discrete Algorithms (SODA)}}, Vol.~\bibinfo{volume}{95}. \bibinfo{pages}{139--149}.
\newblock


\bibitem[Cong and Bader(2005)]%
        {cong2005experimental}
\bibfield{author}{\bibinfo{person}{Guojing Cong} {and} \bibinfo{person}{David Bader}.} \bibinfo{year}{2005}\natexlab{}.
\newblock \showarticletitle{An experimental study of parallel biconnected components algorithms on symmetric multiprocessors (SMPs)}. In \bibinfo{booktitle}{\emph{{IEEE} International Parallel and Distributed Processing Symposium (IPDPS)}}. IEEE.
\newblock


\bibitem[Cormen et~al\mbox{.}(2009)]%
        {CLRS}
\bibfield{author}{\bibinfo{person}{Thomas~H. Cormen}, \bibinfo{person}{Charles~E. Leiserson}, \bibinfo{person}{Ronald~L. Rivest}, {and} \bibinfo{person}{Clifford Stein}.} \bibinfo{year}{2009}\natexlab{}.
\newblock \bibinfo{booktitle}{\emph{Introduction to Algorithms (3rd edition)}}.
\newblock \bibinfo{publisher}{MIT Press}.
\newblock


\bibitem[Dhulipala et~al\mbox{.}(2021)]%
        {gbbs2021}
\bibfield{author}{\bibinfo{person}{Laxman Dhulipala}, \bibinfo{person}{Guy~E. Blelloch}, {and} \bibinfo{person}{Julian Shun}.} \bibinfo{year}{2021}\natexlab{}.
\newblock \showarticletitle{Theoretically efficient parallel graph algorithms can be fast and scalable}.
\newblock \bibinfo{journal}{\emph{{ACM} Transactions on Parallel Computing (TOPC)}} \bibinfo{volume}{8}, \bibinfo{number}{1} (\bibinfo{year}{2021}), \bibinfo{pages}{1--70}.
\newblock


\bibitem[Dhulipala et~al\mbox{.}(2020a)]%
        {dhulipala2020connectit}
\bibfield{author}{\bibinfo{person}{Laxman Dhulipala}, \bibinfo{person}{Changwan Hong}, {and} \bibinfo{person}{Julian Shun}.} \bibinfo{year}{2020}\natexlab{a}.
\newblock \showarticletitle{ConnectIt: a framework for static and incremental parallel graph connectivity algorithms}.
\newblock \bibinfo{journal}{\emph{Proceedings of the VLDB Endowment (PVLDB)}} \bibinfo{volume}{14}, \bibinfo{number}{4} (\bibinfo{year}{2020}), \bibinfo{pages}{653--667}.
\newblock


\bibitem[Dhulipala et~al\mbox{.}(2020b)]%
        {dhulipala2020semi}
\bibfield{author}{\bibinfo{person}{Laxman Dhulipala}, \bibinfo{person}{Charlie McGuffey}, \bibinfo{person}{Hongbo Kang}, \bibinfo{person}{Yan Gu}, \bibinfo{person}{Guy~E Blelloch}, \bibinfo{person}{Phillip~B Gibbons}, {and} \bibinfo{person}{Julian Shun}.} \bibinfo{year}{2020}\natexlab{b}.
\newblock \showarticletitle{Semi-Asymmetric Parallel Graph Algorithms for {NVRAM}s}.
\newblock \bibinfo{journal}{\emph{Proceedings of the VLDB Endowment (PVLDB)}} \bibinfo{volume}{13}, \bibinfo{number}{9} (\bibinfo{year}{2020}).
\newblock


\bibitem[Dinh et~al\mbox{.}(2016)]%
        {dinh2016extending}
\bibfield{author}{\bibinfo{person}{David Dinh}, \bibinfo{person}{Harsha~Vardhan Simhadri}, {and} \bibinfo{person}{Yuan Tang}.} \bibinfo{year}{2016}\natexlab{}.
\newblock \showarticletitle{Extending the nested parallel model to the nested dataflow model with provably efficient schedulers}. In \bibinfo{booktitle}{\emph{{ACM} Symposium on Parallelism in Algorithms and Architectures (SPAA)}}. \bibinfo{pages}{49--60}.
\newblock


\bibitem[Dong et~al\mbox{.}(2021)]%
        {dong2021efficient}
\bibfield{author}{\bibinfo{person}{Xiaojun Dong}, \bibinfo{person}{Yan Gu}, \bibinfo{person}{Yihan Sun}, {and} \bibinfo{person}{Yunming Zhang}.} \bibinfo{year}{2021}\natexlab{}.
\newblock \showarticletitle{Efficient Stepping Algorithms and Implementations for Parallel Shortest Paths}. In \bibinfo{booktitle}{\emph{{ACM} Symposium on Parallelism in Algorithms and Architectures (SPAA)}}. \bibinfo{pages}{184--197}.
\newblock


\bibitem[Dong et~al\mbox{.}(2022)]%
        {bcccode}
\bibfield{author}{\bibinfo{person}{Xiaojun Dong}, \bibinfo{person}{Letong Wang}, \bibinfo{person}{Yan Gu}, {and} \bibinfo{person}{Yihan Sun}.} \bibinfo{year}{2022}\natexlab{}.
\newblock \bibinfo{title}{FAST-BCC: A Parallel Implementation for Graph Biconnectivity}.
\newblock \bibinfo{howpublished}{\url{https://github.com/ucrparlay/FAST-BCC}}.
\newblock


\bibitem[Dong et~al\mbox{.}(2023)]%
        {dong2023provablyfull}
\bibfield{author}{\bibinfo{person}{Xiaojun Dong}, \bibinfo{person}{Letong Wang}, \bibinfo{person}{Yan Gu}, {and} \bibinfo{person}{Yihan Sun}.} \bibinfo{year}{2023}\natexlab{}.
\newblock \showarticletitle{Provably Fast and Space-Efficient Parallel Biconnectivity}.
\newblock \bibinfo{journal}{\emph{arXiv preprint:2301.01356}} (\bibinfo{year}{2023}).
\newblock


\bibitem[Dua and Graf(2017)]%
        {uciml}
\bibfield{author}{\bibinfo{person}{Dheeru Dua} {and} \bibinfo{person}{Casey Graf}.} \bibinfo{year}{2017}\natexlab{}.
\newblock \bibinfo{title}{UCI Machine Learning Repository}.
\newblock \bibinfo{howpublished}{\url{http://archive.ics.uci.edu/ml/}}.
\newblock


\bibitem[Edwards and Vishkin(2012)]%
        {edwards2012better}
\bibfield{author}{\bibinfo{person}{James~A Edwards} {and} \bibinfo{person}{Uzi Vishkin}.} \bibinfo{year}{2012}\natexlab{}.
\newblock \showarticletitle{Better speedups using simpler parallel programming for graph connectivity and biconnectivity}. In \bibinfo{booktitle}{\emph{International Workshop on Programming Models and Applications for Multicores and Manycores (PMAM)}}. \bibinfo{pages}{103--114}.
\newblock


\bibitem[Feng et~al\mbox{.}(2018)]%
        {feng2018distributed}
\bibfield{author}{\bibinfo{person}{Xing Feng}, \bibinfo{person}{Lijun Chang}, \bibinfo{person}{Xuemin Lin}, \bibinfo{person}{Lu Qin}, \bibinfo{person}{Wenjie Zhang}, {and} \bibinfo{person}{Long Yuan}.} \bibinfo{year}{2018}\natexlab{}.
\newblock \showarticletitle{Distributed computing connected components with linear communication cost}.
\newblock \bibinfo{journal}{\emph{Distributed and Parallel Databases}} \bibinfo{volume}{36}, \bibinfo{number}{3} (\bibinfo{year}{2018}), \bibinfo{pages}{555--592}.
\newblock


\bibitem[Fonollosa et~al\mbox{.}(2015)]%
        {fonollosa2015reservoir}
\bibfield{author}{\bibinfo{person}{Jordi Fonollosa}, \bibinfo{person}{Sadique Sheik}, \bibinfo{person}{Ram{\'o}n Huerta}, {and} \bibinfo{person}{Santiago Marco}.} \bibinfo{year}{2015}\natexlab{}.
\newblock \showarticletitle{Reservoir computing compensates slow response of chemosensor arrays exposed to fast varying gas concentrations in continuous monitoring}.
\newblock \bibinfo{journal}{\emph{Sensors and Actuators B: Chemical}}  \bibinfo{volume}{215} (\bibinfo{year}{2015}), \bibinfo{pages}{618--629}.
\newblock


\bibitem[Gu et~al\mbox{.}(2022)]%
        {gu2022parallel}
\bibfield{author}{\bibinfo{person}{Yan Gu}, \bibinfo{person}{Zachary Napier}, \bibinfo{person}{Yihan Sun}, {and} \bibinfo{person}{Letong Wang}.} \bibinfo{year}{2022}\natexlab{}.
\newblock \showarticletitle{Parallel Cover Trees and their Applications}. In \bibinfo{booktitle}{\emph{{ACM} Symposium on Parallelism in Algorithms and Architectures (SPAA)}}. \bibinfo{pages}{259--272}.
\newblock


\bibitem[Gu et~al\mbox{.}(2021)]%
        {gu2021parallel}
\bibfield{author}{\bibinfo{person}{Yan Gu}, \bibinfo{person}{Omar Obeya}, {and} \bibinfo{person}{Julian Shun}.} \bibinfo{year}{2021}\natexlab{}.
\newblock \showarticletitle{Parallel In-Place Algorithms: Theory and Practice}. In \bibinfo{booktitle}{\emph{{SIAM} Symposium on Algorithmic Principles of Computer Systems (APOCS)}}. \bibinfo{pages}{114--128}.
\newblock


\bibitem[Gu et~al\mbox{.}(2015)]%
        {gu2015top}
\bibfield{author}{\bibinfo{person}{Yan Gu}, \bibinfo{person}{Julian Shun}, \bibinfo{person}{Yihan Sun}, {and} \bibinfo{person}{Guy~E. Blelloch}.} \bibinfo{year}{2015}\natexlab{}.
\newblock \showarticletitle{A Top-Down Parallel Semisort}. In \bibinfo{booktitle}{\emph{{ACM} Symposium on Parallelism in Algorithms and Architectures (SPAA)}}. \bibinfo{pages}{24--34}.
\newblock


\bibitem[Hopcroft and Tarjan(1973)]%
        {hopcroft1973algorithm}
\bibfield{author}{\bibinfo{person}{John Hopcroft} {and} \bibinfo{person}{Robert Tarjan}.} \bibinfo{year}{1973}\natexlab{}.
\newblock \showarticletitle{Algorithm 447: efficient algorithms for graph manipulation}.
\newblock \bibinfo{journal}{\emph{Commun. {ACM}}} \bibinfo{volume}{16}, \bibinfo{number}{6} (\bibinfo{year}{1973}), \bibinfo{pages}{372--378}.
\newblock


\bibitem[Hopcroft and Tarjan(1974)]%
        {hopcroft1974efficient}
\bibfield{author}{\bibinfo{person}{John Hopcroft} {and} \bibinfo{person}{Robert Tarjan}.} \bibinfo{year}{1974}\natexlab{}.
\newblock \showarticletitle{Efficient planarity testing}.
\newblock \bibinfo{journal}{\emph{J. {ACM}}} \bibinfo{volume}{21}, \bibinfo{number}{4} (\bibinfo{year}{1974}), \bibinfo{pages}{549--568}.
\newblock


\bibitem[J{\'a}J{\'a}(1992)]%
        {JaJa92}
\bibfield{author}{\bibinfo{person}{Joseph J{\'a}J{\'a}}.} \bibinfo{year}{1992}\natexlab{}.
\newblock \bibinfo{booktitle}{\emph{Introduction to Parallel Algorithms}}.
\newblock \bibinfo{publisher}{Addison-Wesley Professional}.
\newblock


\bibitem[Jamour et~al\mbox{.}(2017)]%
        {jamour2017parallel}
\bibfield{author}{\bibinfo{person}{Fuad Jamour}, \bibinfo{person}{Spiros Skiadopoulos}, {and} \bibinfo{person}{Panos Kalnis}.} \bibinfo{year}{2017}\natexlab{}.
\newblock \showarticletitle{Parallel algorithm for incremental betweenness centrality on large graphs}.
\newblock \bibinfo{journal}{\emph{{IEEE} Transactions on Parallel and Distributed Systems}} \bibinfo{volume}{29}, \bibinfo{number}{3} (\bibinfo{year}{2017}), \bibinfo{pages}{659--672}.
\newblock


\bibitem[Jayanti et~al\mbox{.}(2019)]%
        {jayanti2019randomized}
\bibfield{author}{\bibinfo{person}{Siddhartha Jayanti}, \bibinfo{person}{Robert~E Tarjan}, {and} \bibinfo{person}{Enric Boix-Adser{\`a}}.} \bibinfo{year}{2019}\natexlab{}.
\newblock \showarticletitle{Randomized concurrent set union and generalized wake-up}. In \bibinfo{booktitle}{\emph{{ACM} Symposium on Principles of Distributed Computing (PODC)}}. \bibinfo{pages}{187--196}.
\newblock


\bibitem[Ji and Huang(2020)]%
        {ji2020aquila}
\bibfield{author}{\bibinfo{person}{Yuede Ji} {and} \bibinfo{person}{H~Howie Huang}.} \bibinfo{year}{2020}\natexlab{}.
\newblock \showarticletitle{Aquila: Adaptive parallel computation of graph connectivity queries}. In \bibinfo{booktitle}{\emph{{ACM} International Symposium on High-Performance Parallel and Distributed Computing (HPDC)}}. \bibinfo{pages}{149--160}.
\newblock


\bibitem[Kwak et~al\mbox{.}(2010)]%
        {kwak2010twitter}
\bibfield{author}{\bibinfo{person}{Haewoon Kwak}, \bibinfo{person}{Changhyun Lee}, \bibinfo{person}{Hosung Park}, {and} \bibinfo{person}{Sue Moon}.} \bibinfo{year}{2010}\natexlab{}.
\newblock \showarticletitle{What is Twitter, a social network or a news media?}. In \bibinfo{booktitle}{\emph{International World Wide Web Conference (WWW)}}. \bibinfo{pages}{591--600}.
\newblock


\bibitem[Kwon et~al\mbox{.}(2010)]%
        {cosmo50}
\bibfield{author}{\bibinfo{person}{YongChul Kwon}, \bibinfo{person}{Dylan Nunley}, \bibinfo{person}{Jeffrey~P Gardner}, \bibinfo{person}{Magdalena Balazinska}, \bibinfo{person}{Bill Howe}, {and} \bibinfo{person}{Sarah Loebman}.} \bibinfo{year}{2010}\natexlab{}.
\newblock \showarticletitle{Scalable clustering algorithm for N-body simulations in a shared-nothing cluster}. In \bibinfo{booktitle}{\emph{International Conference on Scientific and Statistical Database Management}}. Springer, \bibinfo{pages}{132--150}.
\newblock


\bibitem[Leskovec et~al\mbox{.}(2009)]%
        {leskovec2009community}
\bibfield{author}{\bibinfo{person}{Jure Leskovec}, \bibinfo{person}{Kevin~J Lang}, \bibinfo{person}{Anirban Dasgupta}, {and} \bibinfo{person}{Michael~W Mahoney}.} \bibinfo{year}{2009}\natexlab{}.
\newblock \showarticletitle{Community structure in large networks: Natural cluster sizes and the absence of large well-defined clusters}.
\newblock \bibinfo{journal}{\emph{im}} \bibinfo{volume}{6}, \bibinfo{number}{1} (\bibinfo{year}{2009}), \bibinfo{pages}{29--123}.
\newblock


\bibitem[Meusel et~al\mbox{.}(2014)]%
        {webgraph}
\bibfield{author}{\bibinfo{person}{Robert Meusel}, \bibinfo{person}{Oliver Lehmberg}, \bibinfo{person}{Christian Bizer}, {and} \bibinfo{person}{Sebastiano Vigna}.} \bibinfo{year}{2014}\natexlab{}.
\newblock \bibinfo{title}{Web Data Commons --- Hyperlink Graphs}.
\newblock \bibinfo{howpublished}{\url{http://webdatacommons.org/hyperlinkgraph}}.
\newblock


\bibitem[Miller and Ramachandran(1987)]%
        {miller1987new}
\bibfield{author}{\bibinfo{person}{G Miller} {and} \bibinfo{person}{Vijaya Ramachandran}.} \bibinfo{year}{1987}\natexlab{}.
\newblock \showarticletitle{A new graphy triconnectivity algorithm and its parallelization}. In \bibinfo{booktitle}{\emph{Proceedings of the nineteenth annual ACM Symposium on Theory of Computing}}. \bibinfo{pages}{335--344}.
\newblock


\bibitem[Miller et~al\mbox{.}(2013)]%
        {miller2013parallel}
\bibfield{author}{\bibinfo{person}{Gary~L Miller}, \bibinfo{person}{Richard Peng}, {and} \bibinfo{person}{Shen~Chen Xu}.} \bibinfo{year}{2013}\natexlab{}.
\newblock \showarticletitle{Parallel graph decompositions using random shifts}. In \bibinfo{booktitle}{\emph{{ACM} Symposium on Parallelism in Algorithms and Architectures (SPAA)}}.
\newblock


\bibitem[Newman and Ghoshal(2008)]%
        {newman2008bicomponents}
\bibfield{author}{\bibinfo{person}{MEJ Newman} {and} \bibinfo{person}{Gourab Ghoshal}.} \bibinfo{year}{2008}\natexlab{}.
\newblock \showarticletitle{Bicomponents and the robustness of networks to failure}.
\newblock \bibinfo{journal}{\emph{Physical review letters}} \bibinfo{volume}{100}, \bibinfo{number}{13} (\bibinfo{year}{2008}), \bibinfo{pages}{138701}.
\newblock


\bibitem[{OpenStreetMap contributors}(2010)]%
        {roadgraph}
\bibfield{author}{\bibinfo{person}{{OpenStreetMap contributors}}.} \bibinfo{year}{2010}\natexlab{}.
\newblock \bibinfo{title}{{OpenStreetMap}}.
\newblock \bibinfo{howpublished}{\url{https://www.openstreetmap.org/}}.
\newblock


\bibitem[Ramachandran(1992)]%
        {ramachandran1992parallel}
\bibfield{author}{\bibinfo{person}{Vijaya Ramachandran}.} \bibinfo{year}{1992}\natexlab{}.
\newblock \bibinfo{booktitle}{\emph{Parallel open ear decomposition with applications to graph biconnectivity and triconnectivity}}.
\newblock \bibinfo{publisher}{Chapter in \emph{Synthesis of Parallel Algorithms}. Morgan Kaufmann.}
\newblock


\bibitem[Reif(1985)]%
        {reif1985depth}
\bibfield{author}{\bibinfo{person}{John~H Reif}.} \bibinfo{year}{1985}\natexlab{}.
\newblock \showarticletitle{Depth-first search is inherently sequential}.
\newblock \bibinfo{journal}{\emph{Inform. Process. Lett.}} \bibinfo{volume}{20}, \bibinfo{number}{5} (\bibinfo{year}{1985}), \bibinfo{pages}{229--234}.
\newblock


\bibitem[Reif(1993)]%
        {Reif93}
\bibfield{author}{\bibinfo{person}{John~H. Reif}.} \bibinfo{year}{1993}\natexlab{}.
\newblock \bibinfo{booktitle}{\emph{Synthesis of Parallel Algorithms}}.
\newblock \bibinfo{publisher}{Morgan Kaufmann}.
\newblock


\bibitem[Sariy{\"u}ce et~al\mbox{.}(2013a)]%
        {sariyuce2013incremental}
\bibfield{author}{\bibinfo{person}{Ahmet~Erdem Sariy{\"u}ce}, \bibinfo{person}{Kamer Kaya}, \bibinfo{person}{Erik Saule}, {and} \bibinfo{person}{{\"U}mit {\c{C} }atalyiirek}.} \bibinfo{year}{2013}\natexlab{a}.
\newblock \showarticletitle{Incremental algorithms for closeness centrality}. In \bibinfo{booktitle}{\emph{IEEE International Conference on Big Data}}. IEEE, \bibinfo{pages}{487--492}.
\newblock


\bibitem[Sariy{\"u}ce et~al\mbox{.}(2013b)]%
        {sariyuce2013shattering}
\bibfield{author}{\bibinfo{person}{Ahmet~Erdem Sariy{\"u}ce}, \bibinfo{person}{Erik Saule}, \bibinfo{person}{Kamer Kaya}, {and} \bibinfo{person}{{\"U}mit~V {\c{C} }ataly{\"u}rek}.} \bibinfo{year}{2013}\natexlab{b}.
\newblock \showarticletitle{Shattering and compressing networks for betweenness centrality}. In \bibinfo{booktitle}{\emph{Proceedings of the 2013 SIAM International Conference on Data Mining}}. SIAM, \bibinfo{pages}{686--694}.
\newblock


\bibitem[Savage and J{\'a}J{\'a}(1981)]%
        {savage1981fast}
\bibfield{author}{\bibinfo{person}{Carla Savage} {and} \bibinfo{person}{Joseph J{\'a}J{\'a}}.} \bibinfo{year}{1981}\natexlab{}.
\newblock \showarticletitle{Fast, efficient parallel algorithms for some graph problems}.
\newblock \bibinfo{journal}{\emph{{SIAM} J. on Computing}} \bibinfo{volume}{10}, \bibinfo{number}{4} (\bibinfo{year}{1981}), \bibinfo{pages}{682--691}.
\newblock


\bibitem[Shen et~al\mbox{.}(2022)]%
        {shen2022many}
\bibfield{author}{\bibinfo{person}{Zheqi Shen}, \bibinfo{person}{Zijin Wan}, \bibinfo{person}{Yan Gu}, {and} \bibinfo{person}{Yihan Sun}.} \bibinfo{year}{2022}\natexlab{}.
\newblock \showarticletitle{Many Sequential Iterative Algorithms Can Be Parallel and (Nearly) Work-efficient}. In \bibinfo{booktitle}{\emph{{ACM} Symposium on Parallelism in Algorithms and Architectures (SPAA)}}.
\newblock


\bibitem[Shun et~al\mbox{.}(2014)]%
        {SDB14}
\bibfield{author}{\bibinfo{person}{Julian Shun}, \bibinfo{person}{Laxman Dhulipala}, {and} \bibinfo{person}{Guy Blelloch}.} \bibinfo{year}{2014}\natexlab{}.
\newblock \showarticletitle{A Simple and Practical Linear-work Parallel Algorithm for Connectivity}. In \bibinfo{booktitle}{\emph{{ACM} Symposium on Parallelism in Algorithms and Architectures (SPAA)}}. \bibinfo{pages}{143--153}.
\newblock


\bibitem[Slota and Madduri(2014)]%
        {slota2014simple}
\bibfield{author}{\bibinfo{person}{George Slota} {and} \bibinfo{person}{Kamesh Madduri}.} \bibinfo{year}{2014}\natexlab{}.
\newblock \showarticletitle{Simple parallel biconnectivity algorithms for multicore platforms}. In \bibinfo{booktitle}{\emph{{IEEE} International Conference on High Performance Computing (HiPC)}}. IEEE, \bibinfo{pages}{1--10}.
\newblock


\bibitem[Tarjan and Vishkin(1985)]%
        {tarjan1985efficient}
\bibfield{author}{\bibinfo{person}{Robert~E Tarjan} {and} \bibinfo{person}{Uzi Vishkin}.} \bibinfo{year}{1985}\natexlab{}.
\newblock \showarticletitle{An efficient parallel biconnectivity algorithm}.
\newblock \bibinfo{journal}{\emph{{SIAM} J. on Computing}} \bibinfo{volume}{14}, \bibinfo{number}{4} (\bibinfo{year}{1985}), \bibinfo{pages}{862--874}.
\newblock


\bibitem[Tsin and Chin(1984)]%
        {tsin1984efficient}
\bibfield{author}{\bibinfo{person}{Yung~H Tsin} {and} \bibinfo{person}{Francis~Y Chin}.} \bibinfo{year}{1984}\natexlab{}.
\newblock \showarticletitle{Efficient parallel algorithms for a class of graph theoretic problems}.
\newblock \bibinfo{journal}{\emph{{SIAM} J. on Computing}} \bibinfo{volume}{13}, \bibinfo{number}{3} (\bibinfo{year}{1984}), \bibinfo{pages}{580--599}.
\newblock


\bibitem[Vishkin(1985)]%
        {vishkin1985efficient}
\bibfield{author}{\bibinfo{person}{Uzi Vishkin}.} \bibinfo{year}{1985}\natexlab{}.
\newblock \showarticletitle{On efficient parallel strong orientation}.
\newblock \bibinfo{journal}{\emph{Inform. Process. Lett.}} \bibinfo{volume}{20}, \bibinfo{number}{5} (\bibinfo{year}{1985}), \bibinfo{pages}{235--240}.
\newblock


\bibitem[Wadwekar and Kothapalli(2017)]%
        {wadwekar2017fast}
\bibfield{author}{\bibinfo{person}{Mihir Wadwekar} {and} \bibinfo{person}{Kishore Kothapalli}.} \bibinfo{year}{2017}\natexlab{}.
\newblock \showarticletitle{A fast GPU algorithm for biconnected components}. In \bibinfo{booktitle}{\emph{International Conference on Contemporary Computing (IC3)}}. IEEE, \bibinfo{pages}{1--6}.
\newblock


\bibitem[Wang et~al\mbox{.}(2021)]%
        {wang2021geograph}
\bibfield{author}{\bibinfo{person}{Yiqiu Wang}, \bibinfo{person}{Shangdi Yu}, \bibinfo{person}{Laxman Dhulipala}, \bibinfo{person}{Yan Gu}, {and} \bibinfo{person}{Julian Shun}.} \bibinfo{year}{2021}\natexlab{}.
\newblock \showarticletitle{GeoGraph: A Framework for Graph Processing on Geometric Data}.
\newblock \bibinfo{journal}{\emph{ACM SIGOPS Operating Systems Review}} \bibinfo{volume}{55}, \bibinfo{number}{1} (\bibinfo{year}{2021}), \bibinfo{pages}{38--46}.
\newblock


\bibitem[Xu et~al\mbox{.}(2022)]%
        {xu2022efficient}
\bibfield{author}{\bibinfo{person}{Yifan Xu}, \bibinfo{person}{Anchengcheng Zhou}, \bibinfo{person}{Grace~Q Yin}, \bibinfo{person}{Kunal Agrawal}, \bibinfo{person}{I-Ting~Angelina Lee}, {and} \bibinfo{person}{Tao~B Schardl}.} \bibinfo{year}{2022}\natexlab{}.
\newblock \showarticletitle{Efficient Access History for Race Detection}. In \bibinfo{booktitle}{\emph{Algorithm Engineering and Experiments (ALENEX)}}. SIAM, \bibinfo{pages}{117--130}.
\newblock


\bibitem[Yang and Leskovec(2015)]%
        {yang2015defining}
\bibfield{author}{\bibinfo{person}{Jaewon Yang} {and} \bibinfo{person}{Jure Leskovec}.} \bibinfo{year}{2015}\natexlab{}.
\newblock \showarticletitle{Defining and evaluating network communities based on ground-truth}.
\newblock \bibinfo{journal}{\emph{Knowledge and Information Systems}} \bibinfo{volume}{42}, \bibinfo{number}{1} (\bibinfo{year}{2015}), \bibinfo{pages}{181--213}.
\newblock


\bibitem[Zheng et~al\mbox{.}(2008)]%
        {geolife}
\bibfield{author}{\bibinfo{person}{Yu Zheng}, \bibinfo{person}{Like Liu}, \bibinfo{person}{Longhao Wang}, {and} \bibinfo{person}{Xing Xie}.} \bibinfo{year}{2008}\natexlab{}.
\newblock \showarticletitle{Learning transportation mode from raw gps data for geographic applications on the web}. In \bibinfo{booktitle}{\emph{International World Wide Web Conference (WWW)}}. \bibinfo{pages}{247--256}.
\newblock


\end{thebibliography}
\end{document}
\endinput